\documentclass[a4paper,USenglish,cleveref, autoref, thm-restate]{lipics-v2021}

\hideLIPIcs  

\title{Geometric $(1+\eps)$-Spanners with Few Crossings} 

\author{Kelvin Luu}{Department of Mathematics, California State University Northridge, Los Angeles, CA, USA}{kelvin.luu.871@my.csun.edu}{https://orcid.org/0009-0001-8182-1278}{}

\author{Csaba D. T\'oth}{Department of Mathematics, California State University Northridge, Los Angeles, CA, USA 
\and Department of Computer Science, Tufts University, Medford, MA, USA}{Csaba.toth@csun.edu}{https://orcid.org/0000-0002-8769-3190}{}

\authorrunning{K. Luu and C. D. T\'oth} 

\Copyright{Kelvin Luu and Csaba D. T\'oth} 

\usepackage[colorinlistoftodos,prependcaption,textsize=small]{todonotes}
\newcommand{\eps}{\varepsilon}
\newcommand{\conv}{\mathrm{conv}}
\newcommand{\diam}{\mathrm{diam}}
\newcommand{\area}{\mathrm{area}}
\newcommand{\width}{\mathrm{width}}
\renewcommand{\angle}{\sphericalangle}

\renewcommand{\angle}{\sphericalangle}

\ccsdesc[500]{Theory of computation~Sparsification and spanners}
\ccsdesc[500]{Theory of computation~Routing and network design problems}
\ccsdesc[500]{Mathematics of computing~Graphs and surfaces}

\keywords{crossings, spanners, empty region property} 

\funding{Research supported in part by the NSF award DMS-2154347.}

\nolinenumbers 

\EventEditors{John Q. Open and Joan R. Access}
\EventNoEds{2}
\EventLongTitle{42nd Conference on Very Important Topics (CVIT 2016)}
\EventShortTitle{CVIT 2016}
\EventAcronym{CVIT}
\EventYear{2016}
\EventDate{December 24--27, 2016}
\EventLocation{Little Whinging, United Kingdom}
\EventLogo{}
\SeriesVolume{42}
\ArticleNo{23}

\begin{document}

\maketitle

\begin{abstract}
For $n$ points in the plane and $\eps>0$, we construct a $(1+\eps)$-spanner with $O(n/\eps)$ edges in which every edge has $\tilde{O}(1/\eps^3)$ crossings, the total number of crossings is $\tilde{O}(n/\eps^4)$, and the ratio between the lengths of any two crossing edges is $O(1/\eps^2)$. 
Our spanner construction substantially improves on the previous upper bound for the number of crossings in a $(1+\eps)$-spanner, and it is the first spanner construction that ensures $O(1)$ crossings per edge for any constant $\eps>0$. In contrast, we construct $n$ points in the plane for which every $(1+\eps)$-spanner has $\Omega(n/\eps^3)$ crossings, $n$ points for which every $(1+\eps)$-spanner has an edge with $\Omega(1/\eps^{5/2})$ crossings, and 4 points for which every $(1+\eps)$-spanner contains two crossing edges where one is $\Omega(1/\eps)$ times longer than the other. 
\end{abstract}

\section{Introduction}
\label{sec:intro}

Given a set $S$ of $n$ points in $\mathbb{R}^d$ and a parameter $t\geq 1$, a \emph{$t$-spanner} is a geometric graph $G=(S,E)$ that contains, for all $a,b\in S$, an $ab$-path of length at most $t\, |ab|$ (where $|ab|$ denotes the Euclidean length of the line segment $ab$). 
The maximum ratio between the shortest-path distance between $a$ and $b$ in $G$ and the Euclidean distance between $a$ and $b$, over all point pairs $a,b\in S$, is called the \emph{stretch} of $G$. 
%
In the plane, the Delaunay triangulation is a noncrossing $1.998$-spanner for every finite set in the plane~\cite{Xia13}, and there is a noncrossing $1.88$-spanner for points in convex position~\cite{BiniazAMSBC16}. However, edge crossings are unavoidable for stretch $t<1.43$ and some point sets in the plane (already in convex position)~\cite{DumitrescuG16}.

Clarkson~\cite{Clarkson87} and independently Keil~\cite{Keil88} showed that every set of $n$ points in the plane admits a $t$-spanner with $O(n)$ edges for any constant $t>1$; this result generalizes to Euclidean $d$-space for constant $d\in \mathbb{N}$~\cite{althofer1993sparse,RS91}. Recent research focused on optimizing the trade-offs between the stretch parameter $t$ and other optimization criteria, such as the number of edges (\emph{sparsity}), total weight (\emph{lightness}), and hop diameter, as $t\to 1$. In this regime, it is convenient to denote the stretch by $t=1+\eps$ for an arbitrarily small $\eps>0$.
For every $\eps>0$, every set of $n$ points in the plane admits a $(1+\eps)$-spanner with $O(n/\eps)$ edges, and this bound is the best possible~\cite{althofer1993sparse,Clarkson87,Keil88}. Many different constructions are known for $(1+\eps)$-spanners: greedy spanners~\cite{althofer1993sparse}, approximate greedy spanners~\cite{DasN97,GudmundssonLN02}, gap-greedy spanners~\cite{AryaS97}, $\Theta$-graphs and Yao-graphs~\cite{Clarkson87,Keil88}, ordered $\Theta$-graphs~\cite{BoseGM04}, as well as spanners based on well-separated pair decomposition (WSPD)~\cite{AryaDMSS95,CallahanK95}, or locality-sensitive orderings (LSOs) \cite{ChanHJ20,GaoH24} and more~\cite{ElkinS15,Toth24}. For a comprehensive overview, refer to the book by Narasimhan and Smid~\cite{NS07}; see also \cite{Smid25}. 

Experiments indicate that greedy spanners often have fewer crossings than other spanner constructions~\cite{FarshiG09}. Motivated by these results, Eppstein and Khodabandeh~\cite{EppsteinK21} studied the crossing pattern of greedy spanners in the plane. They proved that in the greedy $(1+\eps)$-spanner, (i) every edge $e$ crosses $O(1/\eps^8)$ edges that are longer than $e$;
(ii) the total number of crossings is $O(n/\eps^9)$; (iii) however, for every sufficiently small $\eps>0$, there exist $n$-element point sets for which an edge of the greedy $(1+\eps)$-spanner crosses $\Omega(n)$ other edges. 
Their result (ii) immediately implies that the greedy $(1+\eps)$-spanner admits a balanced separator of size $O(\sqrt{n/\eps^9})$; previously, $O(\sqrt{n})$-size balanced separators were known for spanners based on semi-separated pair decompositions (SSPD)~\cite{AbamH12}. 
In higher dimensions, Le and Than~\cite{LeT24} recently showed that for every constant $t>1$, a greedy $t$-spanner for $n$ points in $\mathbb{R}^d$ admits separators of size $O(n^{1-1/d})$. 

We are interested in the minimum number of crossings, and the minimum of the maximum number of crossings per edge in geometric $(1+\eps)$-spanners for $n$ points in the plane. 
This question was also raised recently by Bhore et al.~\cite{BhoreKMTWW26}.

\subsection{Our Results}
\label{ssec:results}
Our main result is a new construction for geometric spanners: Given $n$ points in the plane and $\eps>0$, we construct a $(1+\eps)$-spanner with $O(n/\eps)$ edges such that every edge has $O(\eps^{-3}\log \frac{1}{\eps})$ crossings (\Cref{lem:crossings}) and the total number of crossings is $O(\eps^{-4}\log \frac{1}{\eps}\cdot n)$ (\Cref{cor:crossings}). As a corollary, such a spanner admits a separator of size $O(\eps^{-2} \sqrt{\log(1/\eps)\cdot n})$. 
We also show  that the ratio between the lengths of any two crossing edges in our $(1+\eps)$-spanner is bounded by $O(\eps^{-2})$ (\Cref{lem:ratio}).
The main technical tool is a novel \emph{empty region} property: In our construction, every edge $ab$ is associated with a convex region $R_{ab}$ such that both $a$ and $b$ lie on the boundary $\partial R_{ab}$ and $\mathrm{int}(R_{ab})\cap S=\emptyset$. The region $R_{ab}$ is the convex hull of two congruent regular $\Theta(\eps^{-1})$-gons of diameter $\Theta(\eps\, |ab|)$. Intuitively, the empty region forms a ``buffer'' around the edge $ab$ that prevents shorter edges from crossing $ab$. Even though many geometric graphs have an empty region property (e.g., Delaunay graphs, Gabriel graphs, $\Theta$-graphs, etc.), none of them is known to have few crossings \emph{and} spanning ratio close to 1.

We complement our spanner construction with lower bounds. We show that for every $\eps>0$, there exists a set $S_1$ of $n\geq\Omega(1/\eps)$ points such that every $(1+\eps)$-spanner for $S_1$ has $\Omega(\eps^{-3}n)$ crossings (\Cref{pro:global}); a set $S_2$ of $n\geq \Omega(\eps^{-3/2})$ points such that every $(1+\eps)$-spanner for $S_2$ contains an edge with $\Omega(\eps^{-5/2})$ crossings (\Cref{pro:local}); and a set $S_3$ of 4 points such that every $(1+\eps)$-spanner for $S_3$ contains two crossing edges, one of which is $\Omega(\eps^{-1})$ times longer than the other (\Cref{pro:ratio}). 

\subsection{Empty Region Property}
\label{ssec:emptyregon}
Geometric proximity graphs are intimately related to empty regions: the first known $O(1)$-spanners in the plane were generalized Delaunay graphs~\cite{Chew86,Chew89,RenssenSSW25}, where edges are defined as chords of a homothet of a convex body (e.g., square, triangle, or disk) whose interior is empty (i.e., disjoint from $S$). The spanning ratio of classical Delaunay graphs (defined with empty Euclidean disks) is known to be between $\pi/2$ and 1.998~\cite{Xia13,XiaZ11}. These spanners are plane graphs, but they cannot achieve $1+\eps$ stretch for all $\eps>0$. In $\Theta$-graphs (resp., Yao-graphs), each edge is associated with an empty region, which is an isosceles triangle (resp., sector) with apex angle $\Theta(\eps)$. Unfortunately, an edge associated with the empty region can cross arbitrarily many shorter edges in the neighborhoods of the apex. In our construction, it is essential that the interior angles of the empty region $R_{ab}$ are close to $\pi$ (i.e., nearly flat angles).    

Empty region graphs defined by ``thin'' regions were extensively studied: For example, in the $\beta$-skeleton~\cite{KIRKPATRICK1985217}, an edge $ab$ corresponds to an empty region bounded by two congruent circular arcs with common endpoints $a$ and $b$. A more general family of \emph{empty region} graphs was defined in~\cite{CardinalCL09}: each edge $ab$ corresponds to an empty region, which is homothetic to a given shape with two distinguished points, called anchors, such that the homothety maps the anchors to $a$ and $b$. However, the spanning ratio of $\beta$-skeletons and empty region graphs is known to be unbounded~\cite{BoseDEK06,CardinalCL09,WangLWS03}.

Higher-order proximity graphs relax the empty region condition: edges are defined by regions that contain up to $k$ points; albeit these graphs are not $(1+\eps)$-spanners for small $\eps>0$. {\'{A}}brego et al.~\cite{AbregoMFFHSS11} showed that for $n$ points in the plane, $k$-Delaunay graphs and $k$-Gabriel graphs have $O(k^2n)$ crossings; $k$-nearest neighbor graphs and $k$-relative neighborhood graphs $O(k^3n)$ crossings; and these bounds are the best possible.

\subsection{Technical Highlights}
\label{ssec:highlights}

Given a set $S$ of $n$ points in the plane, we construct a spanning graph $H_2$ on $S$ in three stages (\Cref{sec:spanner}). It will be clear from the construction that $|E(H_2)|=O(n/\eps)$; we analyze the stretch (\Cref{sec:stretch}) and the number of crossings (\Cref{sec:crossing}).  
Here we give a brief high-level overview of our spanner construction (with pointers to the technical definitions and lemmas).

As noted above, Eppstein and Khodabandeh~\cite{EppsteinK21} proved that in a greedy $(1+\eps)$-spanner, every edge $e$ crosses at most $O(1/\eps^8)$ edges longer than $e$.
However, they also showed that there is no upper bound, depending only on $\eps$, on the number of shorter edges that cross $e$. The obstacle is crossings from shorter edges. They constructed point sets for which an edge of the greedy spanner crosses a linear number of shorter edges.

We construct a $(1+\eps)$-spanner by modifying the classical $\Theta$-graph: Unlike greedy spanners, the edges of the $\Theta$-graph are defined locally, and so it is easier to control the impact of modifications. It is known that the $\Theta_k$-graph with  $k=\Theta(\eps^{-1})$ cones is a $(1+\eps)$-spanner~\cite{NS07}. 
Each edge $sp$ of the $\Theta_k$-graph corresponds to an empty region: an isosceles triangle $\Delta(sp)$ with apex $s$, aperture $2\pi/k=\Theta(\eps)$, where $p$ lies on the side opposite to $s$; see \Cref{fig:intro1}. However, an edge $sp$ may have an unbounded number of crossings: It can cross (i) arbitrarily many and arbitrarily short edges close to the apex $s$; and (ii) arbitrarily many edges of length $\Omega(|sp|)$ that traverse the cone of $\Delta(sp)$. We modify the $\Theta$-graph in three stages to control these two types of crossings. 

\begin{figure}[tbh]
        \centering
        \includegraphics[width=.75\textwidth]{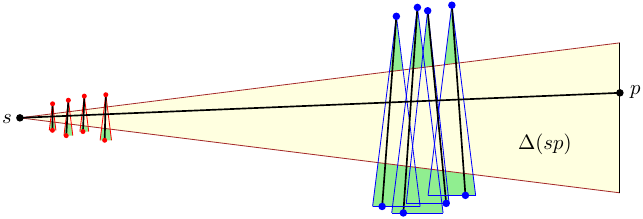}
        \caption{An edge $sp$ in a $\Theta$-graph may have arbitrarily many crossings.}
        \label{fig:intro1}
\end{figure}

In order to control the number of crossings with shorter edges, we replace each edge $sp$ of the $\Theta$-graph with an edge $s'p'$ of similar length and direction that is associated with a ``round'' empty region (i.e., no sharp angles); see \Cref{fig:intro2}. We can show (\Cref{sec:stretch}) that the inductive proof for the stretch $1+\eps$ still goes through with the replacement edges $s'p'$ if we increase the number $k$ of cones by a constant factor. However, the shape of the empty region requires more careful consideration. The first idea for an empty region without sharp angles is the convex hull of two congruent disks, where $s'$ and $p'$ are on the boundaries of the disks (\Cref{fig:intro2}, left). However, this alone would not be sufficient: An edge $s'p'$ could still cross arbitrarily many and arbitrarily short edges 
close to the vertex $s'$ (\Cref{fig:intro2}, right). Instead of disks, we use regular polygons whose sides are aligned with the boundaries of the $k$ cones of the $\Theta$-graph. 
We show (\Cref{lem:ratio}) that this design guarantees that an edge $s'p'$ cannot cross any edges shorter than the side length of the regular polygon, which is $\Theta(\eps^2\, |s'p'|)$.  

\begin{figure}[tbh]
        \centering
        \includegraphics[width=.85\textwidth]{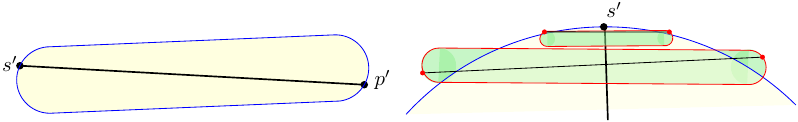}
        \caption{If an empty region of edge $s'p'$ is a union of two congruent disks (left), then an edge $s'p'$ could still cross arbitrarily short edges (right).}
        \label{fig:intro2}
\end{figure}

It remains to handle crossings between $s'p'$ and other edges of length $\Omega(\eps^2 \, |s'p'|)$. Recall that the number of crossings with longer edges was under control for the greedy spanner: It is known that a greedy $(1+\eps)$-spanner cannot contain edges $ab$ and $a'b'$ such that $|aa'|\leq O(\eps \cdot \min\{|ab|,|a'b|\})$ and the angle between $ab$ and $a'b'$ is $O(\eps)$~\cite{ChandraDNS95} (this is known as the \emph{gap property}~\cite{NS07}). 
Even though $\Theta$-graphs do not have this property (cf.~\Cref{fig:intro1}), a pruning process can establish this property at the expense of slightly increasing the stretch. The pruning process (\Cref{ssec:pruning}) proceeds as follows. For each of the $k$ cone directions, we sort the edges of the $\Theta$-graph in increasing order. Then, for each edge $s_1p_1$, we prune all longer edges $s_2p_2$ if $s_1$ and $s_2$ are ``close'', that is, if $s_2$ lies in an \emph{exclusion region} around $s_1$. By default, the exclusion region would be a disk of radius $O(\eps \, |s_1p_1|)$~\cite{GudmundssonLN02}; see \Cref{fig:intro3}~(left). However, a recent analysis for a similar pruning process~\cite{Toth24} shows that the stretch analysis survives even if we use a larger exclusion region: An $O(\eps\, |s_1p_1|)\times O(\sqrt{\eps}\, |s_1p_1|)$ rectangle whose longer side is parallel to the axis of the cone; see \Cref{fig:intro3}~(right).

\begin{figure}[tbh]
        \centering
        \includegraphics[width=.85\textwidth]{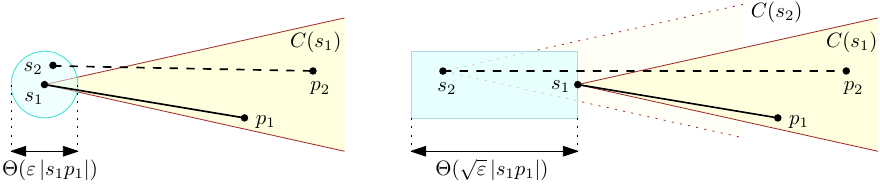}
        \caption{Exclusion regions for pruning. The default choice of a disk of diameter $\Theta(\eps |s_1p_1|)$ (left), and a larger rectangular exclusion region (right).}
        \label{fig:intro3}
\end{figure}

After pruning, we can control the number of crossings between an edge $s'p'$ and edges of length $\Omega(\eps^2 \, |s'p'|)$ using volume arguments (\Cref{lem:long,lem:short}). Larger exclusion regions yield better bounds on the number of crossings. However, pruning raises another technical problem. If $p_1$ is in a cone $C(s_1)$ with apex $s_1$, then the spanner uses edge $s_1p_1$ to reach any other point $t\in C(s_1)$. The edge $s_1p_1$ is helpful for a spanner path from $s_2$ to any point $t\in C(s_1)\cap C(s_2)$, but may not be sufficient for a spanner path to a point $t \in C(s_2) \setminus C(s_1)$. We show (\Cref{ssec:general}) that the stretch analysis from source $s_2$ works for destinations $t$ in a narrower cone $C'(s_2)$ of aperture at most $2\pi/(3k)$ and apex $s_2$. 
We can guarantee stretch $1+\eps$ for all $(s,t)$ pairs if we use three ``copies'' of the modified and pruned $\Theta$-graph, each ``copy'' is obtained by repeating the construction with the initial $k$ cones rotated by $\pm \pi/(3k)$. Then for every $(s,t)$ pair, $t$ lies in a central subcone of aperture $2\pi/(3k)$ within a cone of aperture $2\pi/k$ in one of the three copies. 
For the crossing analysis, the empty regions must be compatible with the boundary cones of all three copies of $\Theta$-graphs:
We define empty regions as the convex hulls of two regular $3k$-gons. 

The modified $\Theta$-graph $H_0$ (defined in \Cref{sec:spanner}) is a $(1+\eps)$-spanner and has the property that the lengths of any two crossing edges $e_1$ and $e_2$, with $|e_1|\leq |e_2|$, satisfy $|e_1|/|e_2|\geq \Omega(\eps^2)$. However, the number of crossings in $H_0$ may be unbounded. After combining three copies and pruning  (\Cref{ssec:pruning}), we obtain a $(1+\eps)$-spanner $H_2$ with $O(\eps^{-3}\log \frac{1}{\eps})$ crossings per edge (\Cref{lem:crossings}), and $O(\eps^{-4}\log \frac{1}{\eps}\cdot n)$ crossings overall (\Cref{cor:crossings}).

\section{Lower Bounds}
\label{sec:LB}

We start with lower bounds for the total number of crossings, for the number of crossings per edge, and for the ratio between lengths of crossing edges in a Euclidean $(1+\eps)$-spanner. Let $S$ be a set of points in the plane. For a point pair $a,b\in S$, let $\mathcal{E}_{ab}$ denote the ellipse with foci $a$ and $b$, and major axis of length $(1+\eps)\, |ab|$. The ellipse $\mathcal{E}_{ab}$ is the set of points $c\in \mathbb{R}^2$ such that $|ac|+|cb|\leq (1+\eps)|ab|$. In particular, if $\mathcal{E}_{ab}\cap S=\{a,b\}$, then $ab$ must be an edge in every $(1+\eps)$-spanner for $S$. We use this property to prove the following three propositions.  

\begin{proposition}\label{pro:global}
     For every $\eps>0$, there exists a set $S$ of $n\geq \Omega(1/\eps)$ points such that every $(1+\eps)$-spanner for $S$ has $\Omega(n/\eps^3)$ crossings.
\end{proposition}
\begin{proof}
We first give a construction for $n=\Theta(1/\eps)$; see \Cref{fig:ellipses}, left. Let $A$ (resp., $B$) be a set of $\lceil 1/(4\eps)\rceil$ points on the top (resp., bottom) side of the unit square $[0,1]^2$, with distance $4\eps$ between consecutive points; and let $S=A\cup B$.
Then for every pair $(a,b)\in A\times B$, we have $\mathcal{E}_{ab}\cap S = \{a,b\}$.
Consequently, every $(1+\eps)$-spanner for $S$ contains the complete bipartite graph with partite sets  $A,B$, which has $\Theta(1/\eps^2)$ edges. By the crossing lemma, the number of crossings of any $(1+\eps)$-spanner $G$ is $\Omega(|E(G)|^3/|V(G)|^2)=\Omega(1/\eps^4) = \Omega(n/\eps^3)$.

For an arbitrary $n\geq \Omega(1/\eps)$, create $\Theta(\eps n)$ copies of the above construction using unit squares at distance at least 1 apart. For each copy and for each point pair $(a,b)\in A\times B$, the property $\mathcal{E}_{ab}\cap S = \{a,b\}$ holds. It follows that each copy generates $\Omega(1/\eps^4)$ crossings, and the total number of crossings is $\Omega(\eps n/ \eps^4)=\Omega(n/\eps^3)$.
\end{proof}

\begin{figure}[tbh]
        \centering
        \includegraphics[width=.95\textwidth]{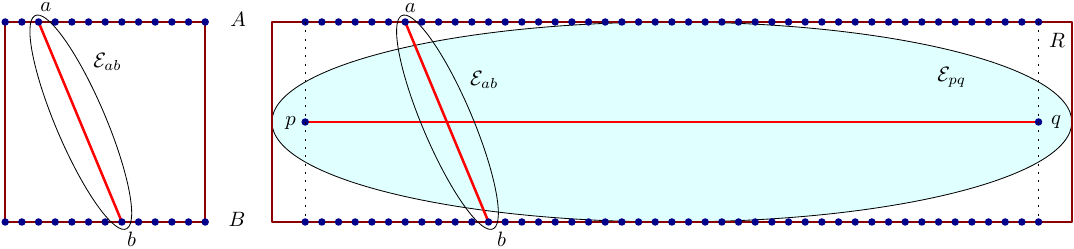}
        \caption{Lower bound constructions}
        \label{fig:ellipses}
\end{figure}

\begin{proposition}\label{pro:local}
     For every $\eps\in (0,1]$, there exists a set $S$ of $n\geq \Omega(1/\eps^{3/2})$ points such that every $(1+\eps)$-spanner for $S$ contains an edge with $\Omega(1/\eps^{5/2})$ crossings.
\end{proposition}
\begin{proof}
Let $p=(0,0)$ and $q=(1,0)$ be two points on the $x$-axis, and let $R$ be the axis-aligned bounding box of the ellipse $\mathcal{E}_{pq}$; see \Cref{fig:ellipses}, right. The height and width of $R$ are $\sqrt{(1+\eps)^2-1}=\sqrt{2\eps+\eps^2} \leq \sqrt{3\eps}$ and $1+\eps$, respectively. 
Let $A$ (resp., $B$) be a set of $\lceil 1/(4\eps^{3/2})\rceil$ points on the top (resp., bottom) side of $R$, with $x$-coordinates in the interval $(0,1)$, with distance $4\eps^{3/2}$ between consecutive points; and let $S=A\cup B\cup \{p,q\}$.

Let $H$ be a $(1+\eps)$-spanner for $S$.
Note that $\mathcal{E}_{pq}\cap S=\{p,q\}$, so $pq$ is an edge of $H$. Furthermore, for every pair $(a,b)\in A\times B$, if $|x(a)-x(b)|\leq \sqrt{3\eps}$, then $\mathcal{E}_{ab}\cap S = \{a,b\}$. Consequently, every point $a\in A$ has degree $\Omega(1/\eps)$ in $H$, and so $H$ contains $\Omega(|A|/\eps)=\Omega(1/\eps^{5/2})$ edges between $A$ and $B$. All edges between $A$ and $B$ cross the edge $pq$, so $pq$ crosses $\Omega(1/\eps^{5/2})$ edges of $H$, as claimed.
\end{proof}

\begin{proposition}\label{pro:ratio}
     For every $\eps\in (0,\frac{1}{16}]$, there is a set $S$ of 4 points such that every $(1+\eps)$-spanner for $S$ has two crossing edges $e_1, e_2$ with  
     $\max\{|e_1|,|e_2|\} /\min\{|e_1|,|e_2|\} \geq \Omega(1/\eps)$.
\end{proposition}
\begin{proof}
Let $S=\{p,q,a,b\}$, where $p=(0,0)$, $q=(1,0)$, $a=(4\eps^{3/2},2\eps)$, and $b=(4\eps^{3/2},-2\eps)$. 
Easy calculation shows that $a,b\notin\mathcal{E}_{pq}$ and $p,q\notin\mathcal{E}_{ab}$. Therefore, any $(1+\eps)$-spanner for $S$ contains both $pq$ and $ab$. These edges cross, and $|pq|/|ab|=1/(4\eps)=\Omega(1/\eps)$, as required.    
\end{proof}

\section{Spanner Construction}
\label{sec:spanner}

Let $S$ be a set of $n$ points in the plane. In this section, we construct a geometric graph $H_2$ on $S$ in three phases. In \Cref{ssec:theta}, we start with a classical $\Theta$-graph for $S$ with $k\geq 6$ cones and $O(kn)$ edges. For each edge $sp$ of the $\Theta$-graph, we construct a new edge $s'p'$ associated with an empty region $R_{s'p'}$ to create $H_0$. In \Cref{ssec:pruning}, we overlay three modified $\Theta$-graphs (with rotated cones) to obtain $H_1$, then prune some edges from $H_1$ to obtain a sparser graph $H_2$. We show that $H_0$, $H_1$, and $H_2$ are $(1+\eps)$-spanners for $S$ for suitable $k=\Theta(1/\eps)$  (\Cref{sec:stretch}); and that every edge of $H_2$ has $O(\eps^{-3}\log \frac{1}{\eps})$ crossings (\Cref{sec:crossing}).

\subsection{Modified $\Theta$-Graphs}
\label{ssec:theta}

We base our construction on the classical $\Theta$-graphs~\cite{NS07}. Let $k$ be an integer. 
Subdivide the plane into $k$ cones of aperture $\alpha = 2\pi/k$ with apex at the origin. 
For each cone $C$, let $\ell_C$ denote the bisector of $C$ emanating from the origin.

Let $S$ be a finite set in the plane. Assume (by a rotation if necessary) that no two points in $S$ span a line parallel to the cone boundaries. 
The $\Theta$-graph $\Theta_k(S)=(S, E)$ is constructed as follows. Initialize $E = \emptyset$. For every $s \in S$, translate the cones and associated rays so that the cones have apex $s$. For each cone $C$ such that $C\cap (S \setminus \{s\} )\neq \emptyset$, let $p\in C$ be a point whose orthogonal projection onto the bisector $\ell_C$ is closest to $s$; add the edge $sp$ to $E$.

\begin{figure}[tbh]
        \centering
        \includegraphics[width=.98\textwidth]{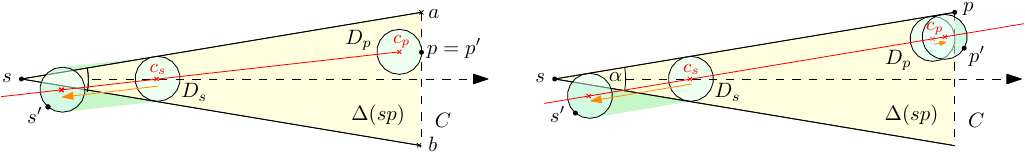}
        \caption{Construction of edge $s'p'$. Left: $p\in D_p$ initially. Right: $p\notin D_p$ initially.}
        \label{fig:construction}
\end{figure}

\subparagraph{Construction.}
Let $S$ be a finite set of points in the plane, and let $G = \Theta_k(S)$ denote the $\Theta$-graph on $S$ constructed with $k=4m+2$ cones of aperture $\alpha=2\pi/k$, for some $m\in\mathbb{N}$. Based on $G$, we construct a new graph $H_0=(S,E(H_0))$ as follows. 
Let $D$ be a regular $3k$-gon such that two opposite sides are parallel to an arbitrary ray between two consecutive cones of the $\Theta$-graph $\Theta_k(S)$. 

For every edge $sp \in E(G)$, we construct a new edge $s'p'$ as follows; see \Cref{fig:construction}. 
Let $C$ be the cone with apex $s$ that contains $p$. Let $\Delta(sp)$ be the isosceles triangle formed by the two rays on the boundary of $C$ and the line through $p$ orthogonal to the bisector of $C$. 
Let $D_s$ and $D_p$ be homothets of $D$ with diameter $\diam(D_s)=\diam(D_p)=\frac{\alpha}{5}\, |sp|$ such that $D_s,D_p \subseteq \Delta(sp)$, and so that $D_s$ is the closest to $s$ and $D_p$ is the closest to $p$ under these constraints. Let $c_s$ and $c_p$, resp., denote the center of $D_s$ and $D_p$, and note that $c_s\neq c_p$. 

Translate $D_s$ and $D_p$, resp., by continuously moving their centers along the rays $\overrightarrow{c_pc_s}$ and $\overrightarrow{c_sc_p}$, resp.,  until their boundary hits a point in $S$. Let $s'\in S$ be a first point hit by $D_s$, and let $p'$ be a first point hit by $D_p$.
We add the edge $s'p'$ to $E(H_0)$. 
This completes the construction of the \emph{modified $\Theta$-graph} $H_0$.

\subparagraph{Proof of Correctness.}
We show that the steps of the construction above are feasible. 
The existence of the distinct homothets $D_s$ and $D_p$ is guaranteed by \Cref{lem:construction:diskradius}.
The termination of the translation step (i.e., the existence of points $s'$ and $p'$) is established by \Cref{lem:termination}.

\begin{restatable}[]{lemma}{constructiondiskradius}
\label{lem:construction:diskradius}
    If $0<\alpha<\frac{28}{25}$, the radius of the inscribed circle of $\Delta(sp)$ is at least $\frac{11\alpha}{50}\, |sa|$. 
\end{restatable}\begin{proof}
    Let $q$ denote the intersection of the edge $ab$ with the angle bisector of $sa$ and $sb$. The segment $sq$ is perpendicular to the segment $ab$, and moreover $|ab| = 2|qa|$.

    The radius of the inscribed circle is given by 
    \[ \varrho = (P - |ab|)\tan\frac{\alpha}{2} = (P- 2|qa|)\tan\frac{\alpha}{2},\]
    where $P = (2|sa| + |ab|)/2 = |sa| + |qa|$ is the semiperimeter of the triangle.
    Then we have
    \begin{align*}
  \varrho &= (P-2|qa|)\tan\frac{\alpha}{2}\\
          &= (|sa| - |qa|)\tan\frac{\alpha}{2}\\
          &= |sa|\left(1-\sin\frac{\alpha}{2}\right)\tan\frac{\alpha}{2} \\
          &> |sa| \left(1- \alpha/2\right)\frac{\alpha}{2} = |sa| \left(\frac{\alpha}{2} - \frac{\alpha^2}{4}\right),
    \end{align*}
    using the Taylor estimates $\sin(x) < x$ and $\tan(x) > x$ valid for $0 < x < \frac{\pi}{2}$.
    An easy calculation shows that if $0 < \alpha < \frac{28}{25}$,
    \[ 
    \varrho > \left(\frac{\alpha}{2} - \frac{\alpha^2}{4}\right)|sa|
    > \frac{11\alpha}{50} |sa|. \qedhere\]
\end{proof}
\begin{restatable}[]{corollary}{corneq}
\label{cor:neq} 
At the initial position of the $3k$-gons $D_p$ and $D_s$, we have $D_p\neq D_s$.
\end{restatable}
\begin{proof}
Let $c$ be the center of the inscribed circle of $\Delta(sp)$. We can find the initial position of $D_s$ (resp., $D_p$) by placing a homothet of $D$ of diameter $\frac{\alpha}{5}\, |sp|$ concentrically with the inscribed circle, and translating it along the ray $\overrightarrow{cp}$ (resp., $\overrightarrow{cs}$) until it hits the boundary of $\Delta(sp)$. 
Since $\overrightarrow{cp}\neq \overrightarrow{cs}$, the translations produce distinct translates $D_p\neq D_s$, as claimed. 
\end{proof}   


Next, we show that the continuous translation of $D_p$ and $D_s$ terminates. We start with an easy observation that justifies the choice of $k\in \mathbb{N}$ and the position of the $3k$-gon $D$.
\begin{lemma}\label{lem:parallel}
    Each side of $\Delta(sp)$ is parallel to two opposite sides of $D$ (hence $D_p$ and $D_s)$. 
\end{lemma}
\begin{proof}
Since $k=4m+2$ for $m\in \mathbb{N}$, all three sides of $\Delta(sp)$ are parallel to a boundary ray of some cone of the $\Theta$-graph~\cite[Lemma~1]{BoseCMRV16}. By the definition of the regular $3k$-gon $D$, every ray on the boundary between consecutive cones is parallel to two opposite sides of $D$. 
\end{proof}


\begin{figure}[tbh]
        \centering
        \includegraphics[width=.98\textwidth]{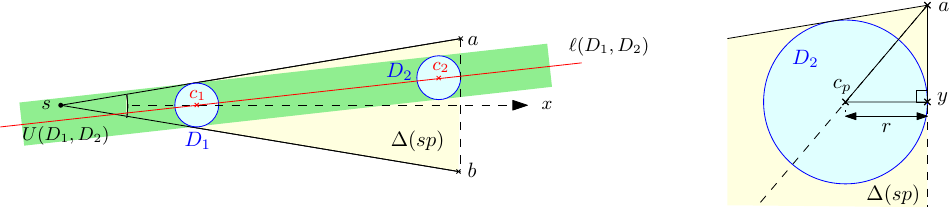}
        \caption{Left: Disks $D_1$ and $D_2$ (blue), the line $\ell(D_1, D_2)$ (red), and the tube $U(D_1, D_2)$ (green). Right: Close-up of the triangle $ac_py$.}
        \label{fig:tube}
\end{figure}

\begin{definition}\label{def:construction:tube}
    Let $D_1$ and $D_2$ be two translates of a centrally symmetric convex body in the plane, with distinct centers $c_1$ and $c_2$, respectively.
    We write $\ell(D_1, D_2)$ for the line passing through the centers $c_1$ and $c_2$, and write $U(D_1, D_2)$ for the tube bounded by the two common outer tangents of $D_1$ and $D_2$; see \Cref{fig:tube}.
\end{definition}

The following lemma pertains to the inscribed disks of $D_p$ and $D_s$.
\begin{lemma}\label{lem:construction:tubegeneric}
    Let $\Delta(sp)$ be an isosceles triangle with vertices $a,b,s$, and let $D_1, D_2 \subseteq \Delta(sp)$ be two disks of the same radius $r < \frac{\alpha}{4}\,|sa|$ with distinct centers $c_1$ and $c_2$, respectively.
    If $D_1$ is tangent to segments $sa$ and $sb$, and $D_2$ is tangent to the segment $ab$, then 
    \begin{enumerate}[(a)]
        \item the tube $U(D_1, D_2)$ contains the point $s$; and
        \item the intersection of $U(D_1, D_2)$ with $ab$ is a line segment of length at least $2r$.
    \end{enumerate}
\end{lemma}
\begin{proof}
    Note that the boundary of $U(D_1, D_2)$ consists of two lines parallel to $\ell(D_1, D_2)$, translated a distance $r$ along the directions orthogonal to $\ell(D_1, D_2)$.
    Because $D_1$ is tangent to both $sa$ and $sb$, then $c_1$ lies on the angle bisector of $\angle asb$. We may assume w.l.o.g.\ that the angle bisector is the positive $x$-axis, and that $c_2$ lies on or above the $x$-axis.

    We observe that if $c_2$ lies on the $x$-axis, then $\ell(D_1, D_2)$ passes through $s$.
    Moreover, as $c_2$ moves up, the distance from $s$ to its orthogonal projection on $\ell(D_1, D_2)$ increases monotonically.
    Thus, it suffices to check that this distance is at most $r$ when $c_2$ is at its maximum height, that is, when $D_2$ is tangent to $sa$. Indeed, as $D_1$ was also tangent to $sa$ and both had radius $r$, by the earlier observation, $sa$ lies on the boundary of $U$, proving (a).

    For (b), note that the endpoints of the intersection of the tube with $ab$ lie on opposite ends of the boundary of $U(D_1, D_2)$.
    Consequently, if $\ell(D_1, D_2)$ makes an angle $\beta$ with $ab$, the length of the intersection is given by $2r\sec\beta$, which is bounded below by $2r$.
\end{proof}

\begin{lemma}\label{lem:termination}
 Let $D_1$ and $D_2$, resp., be the inscribed disks of $D_s$ and $D_p$.
 Then the tube $U(D_1, D_2)$ contains the points $s$ and $p$, and we have $U(D_1, D_2)\subseteq U(D_s,D_p)$.
\end{lemma}
\begin{proof}
   We may assume that the angle bisector of the cone $C$ is the positive $x$-axis. Since $p$ lies on the segment $ab$, the $3k$-gon $D_p$ is tangent to $ab$. 
   Moreover, since two opposite sides of $D_p$ are parallel to $ab$ by \Cref{lem:parallel}, the inscribed disk $D_2$ must be tangent to $ab$ as well. Part (a) of Lemma~\ref{lem:construction:tubegeneric} now implies $s\in U(D_1, D_2)$.

    Assume w.l.o.g.\ that $p$ lies on or above the $x$-axis. We distinguish between two cases.
    If $p$ is sufficiently far away from $a$, then $p\in \partial D_2\subset D_p$.
    In this case, $D_p$ already hits the point $p=p'$ without any translation, so we are done.
    If $p$ is close to $a$, then $D_p$ must also be tangent to $sa$.
    We apply part (b) of Lemma~\ref{lem:construction:tubegeneric}.

    Consider the second case and assume w.l.o.g.\ that $D_2$ is tangent to $sa$; see~\Cref{fig:tube}. In this case, $c_p$ must lie on the angle bisector of $\angle sab$.
    Let $y$ be the intersection point of $D_2$ and $ab$.
    The distance from $c_p$ to $y$ is $r$, and the segment $c_p y$ is orthogonal to $ab$.
    Note that $p$ must lie on the segment $ay$ since otherwise we would be in the setting of the first case.
    Thus, it is enough to compare $|ay|$ to the length of the intersection of $U(D_1, D_2)$ with $ab$.

    Considering the right triangle formed by $a$, $c_p$, and $y$, we obtain
    \[|ay| = \frac{r}{\tan(\frac{\pi - \alpha}{4})} < \frac{4r}{\pi-\alpha} \]
    as $\tan(x) > x$ for $0 < x < \frac{\pi}{2}$.
    Now, an easy calculation shows
    $ |ay| < 4r/(\pi-\alpha) < 2r $
    for $0 < \alpha < \pi-2$, 
    so $ay$ lies within the intersection of $U(D_1, D_2)$ and $ab$.
\end{proof}
When $D_p$ continuously moves towards $p$ (resp., $D_s$ moves towards $s$) in the tube $U(D_s,D_p)$, the translation indeed terminates by the time it reaches $p$ (resp., $s)$.

\subparagraph{Empty Region Property.}
We have shown that for every edge $sp$ of the $\Theta$-graph $G$ (equivalently, for every nonempty cone $C$ with apex $s$), $H_0$ contains a corresponding edge $s'p'$.
The construction also guarantees the following empty region property for every $s'p'\in E(H_0)$. Let $R_{s'p'}=\conv(D_s\cup D_p)$, where $D_s$ and $D_p$ are the $3k$-gons in their final position. 

\begin{lemma}\label{lem:emptyregion}
    For every $s'p'\in E(H_0)$, we have $\mathrm{int}(R_{s'p'})\cap S=\emptyset$.
\end{lemma}
\begin{proof}
The construction of the $\Theta$-graph guarantees that $\mathrm{int}(\Delta(sp))\cap S=\emptyset$. In their initial position, both $D_s$ and $D_p$ are contained in $\Delta(sp)$. At that time, we have $\conv(D_s\cup D_p)\subset \Delta(sp)$, 
hence $\mathrm{int}(\conv(D_s\cup D_p))\cap S=\emptyset$. When we continuously translate $D_p$ and $D_s$, then $\conv(D_s\cup D_p)$ increases by a region swept by $D_p$ and $D_s$. Since they do not sweep through any point in $S$, the continuous translation maintains the invariant $\mathrm{int}(\conv(D_s\cup D_p))\cap S=\emptyset$, and it implies $\mathrm{int}(R_{s'p'})\cap S=\emptyset$ at the end of the translation.
\end{proof}

The following lemma records the properties of the two cases mentioned in the proof of \Cref{lem:termination}: These properties will be used in the crossing analysis (\Cref{sec:crossing}).

\begin{lemma}\label{lem:cases}
For every $s'p'\in E(H_0)$, let $D_s$ and $D_p$ be the $3k$-gons in their final position. At least one of the following statements holds:
\begin{enumerate}
    \item the two lines on the boundary of $U(D_p,D_s)$ are parallel to a ray on the boundary of $C$;
    \item $p=p'$ and $D_p\subset \Delta(sp)$. 
\end{enumerate}
\end{lemma}
\begin{proof} 
By \Cref{lem:parallel}, each side of $\Delta(sp)$ is parallel to two opposite sides of $D_p$. 
At its initial position, one side of $D_p$ is contained in the side $ab$ of $\Delta(sp)$.
If $p\in \partial D_p$ at this position, then $p=p'$ and $D_p\subset \Delta(sp)$. Otherwise, $D_p$ is initially tangent to two sides of $\Delta(sp)$: $ab$ and either $sa$ or $sb$. As $D_s$ is initially tangent to $sa$ and $sb$, both $D_s$ and $D_p$ are tangent to a boundary ray $\overrightarrow{r}$ of the cone $C$. In this case, the line $c_sc_p$ is parallel to $\overrightarrow{r}$. Both outer common tangents of $D_s$ and $D_p$ are parallel to $\overrightarrow{r}$ before translation, and these lines form the boundary of the tube $U(D_s,D_p)$. The translation does not change the tube $U(D_s,D_p)$.
\end{proof}

\subsection{Overlay and Pruning}
\label{ssec:pruning}

\subparagraph{Overlay of Three $\Theta$-Graphs.}
In \Cref{ssec:theta}, we have constructed a modified $\Theta$-graph $H_0$ from the classical $\Theta$-graph with $k$ cones of aperture $\alpha=2\pi/k$. 
If we rotate the $k$ cones counterclockwise through $\frac13 \alpha$ (resp., $\frac23\alpha$), this construction yields another modified $\Theta$-graph $H_0'$ (resp., $H_0'')$. Note that the boundaries of the cones of all three $\Theta$-graphs are parallel to some of the sides of the regular $3k$-gon $D$. 
(This is why we define $D$ to be a regular $3k$-gon; the stretch analysis of $H_0$ in \Cref{ssec:stretch1} would go through even if $D$ is a regular $k$-gon.)
Let $H_1$ be the union of these three modified $\Theta$-graphs, that is, let $E(H_1)=E(H_0)\cup E(H_0')\cup E(H_0'')$. 

\subparagraph{Pruning.} 
We define a subgraph $H_2\subseteq H_1$. Recall that the edges $s'p'\in E(H_1)$ correspond to edges $sp$ of three $\Theta$-graphs. We first prune some of the edges from these $\Theta$-graphs, and define $E_2$ as the set of edges $s'p'$ corresponding to the surviving edges $sp$. 

We begin by classifying the edges of the three $\Theta$-graphs by direction and length, and then prune the edges in each class independently. 
Assume that the union of the three $\Theta$-graphs corresponds to cones $\mathcal{C}=\{C_1,\ldots , C_{3k}\}$ in an arbitrary order (where each $\Theta$-graph contributes $k$ cones, and the cones may overlap), and let $C_i(s)$ be the translate of $C_i$ with apex $s$ for all $i\in \{1,\ldots, 3k\}$.
Let $\varphi=\frac54$. For every $i\in \{1,\ldots, 3k\}$ and every $j\in \mathbb{Z}$, let $E_{i,j}$ be the set of edges $sp$ of the $\Theta$-graphs corresponding to a cone $C_i(s)$ 
for some $s\in S$  
with length $\varphi^{j-1}\leq |sp|<\varphi^j$. 

For fixed $i,j$, we describe the pruning process for $E_{i,j}$. Assume (by applying a rotation for the sake of the pruning process) that the angle bisector of the cone $C_i$ is the positive $x$-axis. For every edge $sp\in E_{i,j}$, we define an \emph{exclusion region}  (see \Cref{fig:exclusion} for an illustration):
\[
    A(sp) =\left[x(s) -\frac{1}{10}|sp|,x(s)\right] \times \left[y(s)-\frac{\alpha}{20}|sp|, y(s)+\frac{\alpha}{20}|sp|\right].
\]

\begin{figure}[tbh]
        \centering
        \includegraphics[width=.75\textwidth]{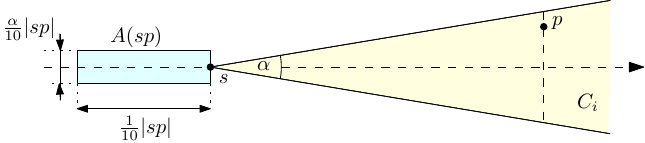}
        \caption{The exclusion region $A(sp)$ for an edge $sp$ (the figure is not to scale)}
        \label{fig:exclusion}
\end{figure}

We prune some of the edges of $E_{i,j}$ as follows. Sort the edges in $E_{i,j}=\{s_1p_1,\ldots , s_m p_m\}$ such that $x(s_\ell)$ is monotonically nonincreasing. Initialize $E^*_{i,j}:=E_{i,j}$. For $h=1$ to $m$, if $s_h p_h\in E^*_{i,j}$, then remove from $E^*_{i,j}$ all other edges $s_g p_g$ such that $s_g\in A(s_h p_h)$. 

Finally, we construct the graph $H_2=(S,E(H_2))$: Initialize $E(H_2)=\emptyset$. For all $i,j$, and all $sp\in E_{i,j}^*$, add the edge $s'p'$ to $E(H_2)$. In \Cref{sec:stretch}, we show that $H_2$ is a $(1+\eps)$-spanner for $S$; and in \Cref{sec:crossing} we show that every edge of $H_2$ has $O(\eps^{-3}\log \frac{1}{\eps})$ crossings. For convenience, let $E_{i,j}(H_2)$ denote the set of edges in $E(H_2)$ associated with edges in $E^*_{i,j}$.

\subparagraph{Packing Property.}
Consider the set of edges $E^*_{i,j}$ at the end of the pruning process. The associated exclusion regions $A(sp)$, $sp\in E^*_{i,j}$, are not necessarily disjoint; see \Cref{fig:overlap}. However, their intersection is limited in some sense. Specifically, for every $sp\in E^*_{i,j}$, let $B(sp)$ be obtained from $A(sp)$ by a dilation of factor $\frac13$ with center $s$; see \Cref{fig:overlap}.

\begin{restatable}[]{lemma}{lempacking}
\label{lem:packing}
 For all $i,j$, the rectangles $B(sp)$ for all edges $sp\in E^*_{i,j}$ are disjoint. 
\end{restatable}
\begin{figure}[tbh]
        \centering
        \includegraphics[width=.55\textwidth]{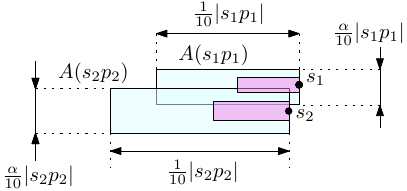}
        \caption{The rectangles $A(s_1p_1)$ and $A(s_2p_2)$ overlap, but $B(s_1p_1)$ and $B(s_2p_2)$ are disjoint.}
        \label{fig:overlap}
\end{figure}
\begin{proof}
    Suppose, for the sake of contradiction that there exists $i,j$ and two distinct edges $s_1p_1,s_2p_2\in E^*_{i,j}$ such that $B(s_1p_1)\cap B(s_2p_2)\neq \emptyset$. Since $B(s_1p_1)\subset A(s_1p_1)$ and $B(s_2p_2)\subset A(s_2p_2)$, then $A(s_1p_1)\cap A(s_2p_2)\neq \emptyset$. Assume w.l.o.g.\ that the angle bisector of the cone $C_i$ is the positive $x$-axis and $x(s_2)\leq x(s_1)$. We may also assume  (by a reflection in the $x$-axis if necessary) that $y(s_1)\leq y(s_2)$; see \Cref{fig:overlap}.

    Since $x(s_2)\leq x(s_1)$, then the pruning process has considered edge $s_1p_1$ before $s_2p_2$. Therefore $x(s_1)-\frac{1}{10}|s_1p_1|\leq x(s_2)\leq x(s_1)$. 
    Since $s_2p_2$ has not been removed in pruning, then $s_2\notin A(s_1p_1)$. In particular, this means that $|y(s_1)-y(s_2)|>\frac{\alpha}{20}|s_1p_1|$. Recall that every edge $e\in E_{i,j}$ satisfies $\varphi^{j-1}\leq |e|<\varphi^j$, where $\varphi=\frac54$. Consequently $|s_1p_1|< \frac54\, |s_2p_2|$. The rectangle $B(s_1p_1)$ lies below the horizontal line $L:y=y(s_1)+\frac{\alpha}{60}|s_1p_1|$, and rectangle $B(s_2p_2)$ lies above the horizontal line 
    \begin{align*}
    y&=y(s_2)-\frac{\alpha}{60}|s_2p_2|
    = y(s_1)+ |y(s_1)-y(s_2)|- \frac{\alpha}{60}|s_2p_2|\\
    &>y(s_1)+\frac{\alpha}{20}|s_1p_1|-\frac{\varphi\, \alpha}{60} |s_1p_1|
    =y(s_1)+\left(\frac{1}{20}-\frac{1}{48}\right)\alpha |s_1p_1|\\
    &=y(s_1)+\frac{7\alpha}{240}|s_1p_1| 
    >y(s_1)+\frac{\alpha}{60}|s_1p_1|.
    \end{align*}
    This shows that $B(s_1p_1)$ and $B(s_2p_2)$ are on opposite sides of the horizontal line $L$, which is a contradiction.
    \end{proof}

\section{Stretch Analysis}
\label{sec:stretch}

In \Cref{ssec:stretch1} we show that $H_0$ (before the pruning phase) is a $(1+\eps)$-spanner if $\eps/32\leq \alpha\leq\eps/16$; and in \Cref{ssec:general}, 
we show that $H_2$ is a $(1+\eps)$-spanner if $\eps/32\leq \alpha\leq \eps/16$. Throughout this section, we use the Taylor estimates
\begin{equation}\label{eq:Taylor}
    \frac{23x}{24}<x-\frac{x^3}{3!}<\sin x<x, \hspace{5mm} 1-\frac{x^2}{2}<\cos x<1, \hspace{5mm} x<\tan x < x+2\cdot \frac{x^3}{3}<\frac{7x}{6},
\end{equation}
for $0<x<\frac12$; and assume that $0<\alpha<\frac12$. The Euclidean length of every line segment $ab$ is bounded above by its $L_1$-norm, that is, $|ab|\leq |x(a)-x(b)|+|y(a)-y(b)|$.

\subsection{Graph $H_0$ is a $(1+\eps)$-spanner}
\label{ssec:stretch1}

We prove that the graph $H_0$ (before the pruning phase) is a $(1+\eps)$-spanner. 
Let $s, t\in S$ be two distinct points. Let $C$ be the cone of the $\Theta$-graph with apex $s$ and aperture $\alpha$ that contains $t$, and let $\ell_C$ be the angle bisector of $C$ emanating from $s$.
As $C\cap (S\setminus \{s\})\neq \emptyset$,  the $\Theta$-graph has an edge $sp \in E(G)$ in $C$, and we have also constructed an edge $s'p' \in E(H_0)$.
We would like to construct an $st$-path in $H_0$ by concatenating a shortest path $P(s, s')$ from $s$ to $s'$ in $H_0$ together with the edge $s'p'$ and a shortest path $P(p', t)$ from $p'$ to $t$ in $H_0$. To that end, we prove the following lemmas to obtain the paths $P(s, s')$ and $P(p', t)$ by induction.

\begin{restatable}[]{lemma}{lemapexcenter}
\label{lem:apexcenter}
    The center $c_s$ of $D_s$ in its initial position satisfies $|sc_s| < \frac14\, |sp|$.
\end{restatable}
\begin{proof}
    Recall that in its initial position, $D_s$ has half-diameter $r=\frac{\alpha}{10}\, |sp|$ and must be tangent to both rays of the boundary of $C$. 
    Consequently, $c_s$ lies on the angle bisector of $C$, $\ell_C$.
    Let $w$ be the point at which (the inscribed disk of) $D_s$ intersects one of the boundary rays, and consider the right triangle formed by $s$, $c_s$, and $w$.
    We know $\angle c_s s w  \leq \alpha/2$ as $c_s$ lies on the angle bisector $\ell_C$, and moreover $|c_s w|\leq \frac{\alpha}{10}\,|sp|$.
    Therefore, we have
    \[ |sc_s| 
    \leq \frac{\alpha/10}{\sin(\alpha/2)}|sp| 
    \leq \frac{\alpha/10}{23\alpha/48} |sp| 
    =\frac{24}{115}\, |sp|
    < \frac14\, |sp|.\]
    using Taylor estimates (cf.~\Cref{eq:Taylor}).
\end{proof}

It will be convenient in the stretch analysis to define coordinate axes by taking $s$ to be the origin and the ray $\overrightarrow{st}$ to be the positive $x$-axis. Observe that in this coordinate system, $|x(s)| = 0$, and $|y(s)|=|y(t)| = 0$, and $|st|=|x(s)-x(t)| = |x(t)|$. We introduce a few useful auxiliary points and bounds on their $x$-coordinates, and then proceed to basic computations of the $x$ and $y$-coordinates of $s'$ and $p'$ necessary for the remaining argument.

\begin{figure}[tbh]
    \centering
    \includegraphics[width=0.8\linewidth]{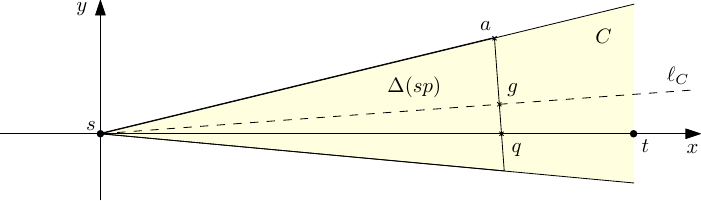}
    \caption{Auxiliary points $q$, $a$ and $g$.}
    \label{fig:auxpoints}
\end{figure}

Recall that $\Delta(sp)$ is the isosceles triangle formed by the two rays on the boundary of $C$ and the line through $p$ orthogonal to the angle bisector $\ell_C$. We call the side of $\Delta(sp)$ on the line through $p$ orthogonal to $\ell_C$ the \emph{third side of} $\Delta(sp)$.

\begin{lemma}\label{lem:auxpoints}
Let $q$ be the intersection of $st$ with the third side of $\Delta(sp)$, and let $a$ be the vertex of $\Delta(sp)$ farthest from the $x$-axis; see~\Cref{fig:auxpoints}. Then, $x(t) \geq x(q) \geq x(a) \geq \frac78\, |sp|$.
\end{lemma}
\begin{proof}
    We notice first that $q$ exists as $t \in C$. By the construction of the $\Theta$-graph, we have $S\cap \mathrm{int}(\Delta(sp))=\emptyset$ and thus $x(t) \geq x(q)$. Moreover, by choice of $a$, $x(a)$ is minimal among all points on the third side, so that $x(q) \geq x(a)$. Finally, 
    since $\angle asq \leq \alpha < \frac12$ and $|sp| \leq |sa|$, Taylor estimates (cf.~\Cref{eq:Taylor}) yield
    \[
    x(a) = |sa|\cos(\angle asq) 
    \geq |sa|\cos\alpha 
    \geq |sp|\left(1- \frac{\alpha^2}{2}\right) 
    \geq \frac78\, |sp|.
    \qedhere
    \]
\end{proof}

\begin{lemma}\label{lem:xycoords}
 We have the following bounds on the $x$ and $y$-coordinates of $s'$ and $p'$.
\begin{enumerate}
\item $\frac{-\alpha}{5}\, |sp| \leq x(s') \leq (\frac14 + \frac{\alpha}{10})|sp|$ and $|y(s')| \leq \frac{\alpha}{5}\, |sp|$.
\item $\frac78\, |sp| \leq x(p') \leq (1 + \frac{\alpha}{5}) |sp|$ and $|y(p')| < \frac{4\alpha}{3}\, |sp|$.
\item If $x(t) \leq x(p')$, then $x(p') - x(t) < \frac{2\alpha}{5}\, |sp|$, otherwise, $x(t) - x(p') \leq |st| - \frac78\, |sp|$.
\end{enumerate}
\end{lemma}
\begin{proof}
    We begin by proving the $x$-coordinate inequalities.
    Observe first that the $x$-coordinates of points on $D_s$ are nonincreasing during the translation from its initial position. Similarly, the $x$-coordinates of points on $D_p$ are nondecreasing during the translation from its initial position.

    By monotonicity of the $x$-coordinates of $D_s$, $x(s')$ is minimal when $s$ lies on the boundary of $D_s$. 
    Since the $x$-projection of $D_s$ is an interval of length at most $\diam(D_s)=\frac{\alpha}{5} \,|sp|$, we have that $x(s') \geq x(s) - \frac{\alpha}{5}\, |sp| = \frac{-\alpha}{5}\, |sp|$.
    Similarly, we have that $x(p') \leq x(p) + \frac{\alpha}{5}\,|sp| \leq (1+\frac{\alpha}{5})|sp|$.

    For the upper bound on $x(s')$, the maximal $x$-coordinate of any point in $D_s$ in its initial position is at most half a diameter more than $x(c_s) = x(c_s) - x(s)$, so that by Lemma~\ref{lem:apexcenter},
    \[ 
    x(s') 
    \leq x(c_s) + \frac{\alpha}{10}\, |sp| 
    \leq |sc_s| + \frac{\alpha}{10}\, |sp| 
    \leq \left(\frac14+\frac{\alpha}{10}\right)|sp|.
    \]

    For the lower bound on $x(p')$, we observe that $x(p') \geq x(a) \geq \frac78\, |sp|$ by Lemma~\ref{lem:auxpoints}.

    Now, for the distance between $x(p')$ and $x(t)$, notice first that
    we must have $x(q) \leq x(t)$ as $S\cap \mathrm{int}(\Delta(sp))=\emptyset$.
    In the case where $x(t) \leq x(p')$, we can bound $x(p') - x(t) \leq x(p') - x(q)$ as follows.
    Let $g$ be the intersection of $\ell_C$ with the third side, and let $\beta = \angle gsq$ be the angle $st$ makes with $\ell_C$.
    As $x(p') \leq x(p) +\frac{\alpha}{5}\,|sp|$, $x(p') -x(q) \leq x(p) - x(q) + \frac{\alpha}{5}\,|sp|$, so we need an upper bound for $x(p) - x(q)$.
    Supposing that $x(p)\geq x(q)$, $p$ cannot lie between $g$ and $q$.
    Consequently, $x(p)-x(q) = |pq|\sin\beta \leq |pq|\sin\frac{\alpha}{2} \leq \frac{\alpha}{2}\, |pq|$, and
    \[|pq| 
    = |sg|\left(\tan\frac{\alpha}{2}-\tan\beta\right) 
    \leq |sg| \tan\frac{\alpha}{2} 
    \leq \frac{7\alpha}{12}\, |sg| 
    \leq \frac{7\alpha}{12}\, |sp|\]
    by Taylor estimates (cf.~\Cref{eq:Taylor}) and since $|sg| \leq |sp|$. Thus, $x(p)-x(q) \leq \frac{7\alpha^2}{24}\, |sp| \leq \frac{7\alpha}{48}\, |sp|$ as $\alpha < \frac12$.
    We have
    \[ 
    x(p') - x(t) 
    \leq x(p') - x(q) 
    \leq x(p) - x(q) + \frac{\alpha}{5}\, |sp| 
    \leq \left(\frac{7\alpha}{48} +\frac{\alpha}{5}\right) |sp|
    = \frac{83\alpha}{240}\, |sp|
    < \frac{2\alpha}{5}\, |sp|.
    \]
    Suppose instead that $x(t) \geq x(p')$. We have $x(t) - x(p') = |st| - x(p') \leq |st| - \frac78\, |sp|$.

    For the $y$-coordinates, we first find an upper bound for $|y(p')|$.
    Since $p'$ and $t$ lie in $C$,
    \[|
    y(p')|= |x(p')| \tan(\angle p'st) 
    \leq \left(1+\frac{\alpha}{5}\right)|sp| \tan\alpha
    < \frac{11}{10}\cdot \frac{7\alpha}{6} \, |sp|
    = \frac{77\alpha}{60}\, |sp|
    <\frac{4\alpha}{3}\, |sp| 
    \]
    using the assumption $\alpha < \frac12$ and Taylor estimates (cf.~\Cref{eq:Taylor}).

    Next, for $s'$, we show that $|y(s')| \leq \diam(D_s) = \frac{\alpha}{5}\, |sp|$. The $y$-coordinates of $D_s$ are monotone during the translation, and moreover by Lemma~\ref{lem:termination}, $s$ must eventually be hit by $D_s$ after translating far enough. But if some point $q$ satisfies $|y(q)|> \diam(D_s) = \frac{\alpha}{5}\, |sp|$ and lies on $D_s$, then no point of $D_s$ has $y$-coordinate $0$. Consequently, monotonicity ensures $s$ cannot be hit, a contradiction.
\end{proof}

Using the bounds from Lemma~\ref{lem:xycoords} and the assumption $\alpha<\frac12$, we obtain
\begin{corollary}\label{cor:spprime}
$|x(s') - x(p')|  >\frac{23}{40}\, |sp|.$
\begin{proof}
$       |x(s') - x(p')| 
   \geq \frac78\, |sp|  -  \left(\frac14 + \frac{\alpha}{10}\right)|sp|
   = \left(\frac58 -\frac{\alpha}{10}\right)|sp|
   >\frac{23}{40}\, |sp| $.
\end{proof}
\end{corollary}

\begin{lemma}\label{lem:inductionhyp}
If $\alpha < \frac12$, we have $|ss'| < |st|$ and $|p't| < |st|$.
\end{lemma}
\begin{proof} 
The Euclidean length of the segment $ss'$ is bounded above by its $L_1$ norm. 
\Cref{lem:xycoords} combined with $\alpha < \frac12$ implies that
     \[
      |ss'| \leq |x(s')| + |y(s')| 
      \leq \left(\frac14+\frac{\alpha}{10}\right)|sp| + \frac{\alpha}{5}\, |sp| 
      = \left(\frac14+\frac{3\alpha}{10}\right)|sp|
      <\frac{2}{5}\, |sp| <\frac78\, |sp| <|st| .
      \]
    For the segment $p't$, first assume that $x(t) \geq x(p')$. \Cref{lem:xycoords} combined with  $\alpha<\frac12$ gives
    \[ 
    |p't| \leq   x(t) - x(p') + |y(p')| 
    \leq |st| - \frac78\, |sp| + \frac{4\alpha}{3}\, |sp| 
    < |st| + \left(\frac{2}{3} - \frac78\right)|sp|
    <|st| ,
    \]
   Similarly, if $x(t) \leq x(p')$,  \Cref{lem:xycoords} combined with $\alpha < \frac12$ yields
     \[ 
    |p't| \leq x(p') - x(t) + |y(p')| 
     \leq \left(\frac{2}{5} + \frac{4}{3}\right)\alpha|sp| 
     = \frac{26\alpha}{15}\, |sp|
     <\frac{13}{15} |sp|
     <\frac78\, |sp|
     < |st|
     .\qedhere
     \]
\end{proof}

We can now prove the main result of \Cref{ssec:stretch1}.
\begin{lemma}\label{lem:stretchH0}
    If $\eps < 1$, the graph $H_0$ 
    with cones of aperture $\alpha \leq \eps/16$ is a $(1+\eps)$-spanner.
\end{lemma}
\begin{proof}
    We show that $H_0$ contains a path $P(s,t)$ of length $|P(s,t)| \leq (1+\eps)|st|$. 
    We proceed by induction on point pairs $s, t \in S$ in increasing order of their Euclidean distance $|st|$. In the base case, we have $s=t$, so the claim trivially holds.
    
    Now consider a pair $s, t \in S$ with $s\neq t$, and assume the claim holds for all pairs $\{s^*, t^*\} \subseteq S$ with $|s^* t^*| < |st|$. We know $t$ lies in some cone $C$ with apex $s$, and the construction of $H_0$ yields some edge $s'p'$ associated to $C$. By \Cref{lem:inductionhyp} and the induction hypothesis, $H_0$ contains paths $P(s, s')$ and $P(p', t)$ satisfying $|P(s,s')| \leq (1+\eps)|ss'|$ and $|P(p', t)| \leq (1+\eps)|p't|$. 
    We consider the $st$-path $P(s, t)$ obtained by concatenating $P(s, s')$ with the segment $s'p'$ and $P(p', t)$.
    Assume $s$ is the origin and the ray $\overrightarrow{st}$ is the positive $x$-axis; see \Cref{fig:stretchH0}. The length of any segment $ab$ is bounded above by its $L_1$ norm, i.e., $|ab| \leq |x(a) - x(b)| + |y(a) + y(b)|$. We proceed by finding upper bounds for distances in $x$ and $y$-projection.
 
\begin{figure}[tbh]
        \centering
        \includegraphics[width=.75\textwidth]{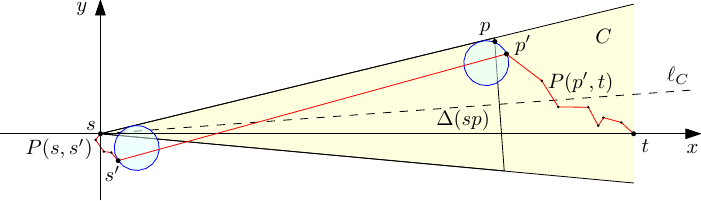}
        \caption{The concatenation of $P(s, s')$, $s'p'$ and $P(p', t)$}
        \label{fig:stretchH0}
\end{figure}
   


   \subparagraph{Upper bounds for distances in $x$-projections.} First, we compute the total $x$-distance traveled along the path. Although the path begins at $s$ and terminates at $t$, the path could backtrack if $x(s') < x(s)$ or $x(p') > x(t)$. We can bound the backtracking by Lemma~\ref{lem:xycoords}:
    \[ |x(s) - x(s')| + |x(s')-x(p')| + |x(p')-x(t)| \leq |st| + 2\left(\frac{\alpha}{5}\,|sp| + \frac{2\alpha}{5}\,|sp|\right) = |st| + \frac{6\alpha}{5}\,|sp|.\]
    In particular, we have that
    \[|x(s) - x(s') + |x(p') - x(t)| \leq |st| + \frac{6\alpha}{5}\,|sp| - |x(s')-x(p')|.\]

    \subparagraph{Upper bounds for distances in $y$-projections.} Using the bounds in \Cref{lem:xycoords} and the fact that $y(s)=y(t)=0$, the triangle inequality yields
    \[  |y(s) - y(s')|+ |y(s') - y(p')| + |y(p') - y(t)| 
    \leq 2|y(s')| + 2|y(p')| 
    \leq \frac{2\alpha}{5}\,|sp| + \frac{8\alpha}{3}\, |sp| 
    = \frac{46\alpha}{15}\, |sp|.\]

    \subparagraph{Length of $P(s,t)$.}
    We bound the length of the path $P(s,t)$.
    By the induction hypothesis,
    \begin{align*}
    |P(s,t)| &= |P(s, s')| + |s'p'| + |P(p', t)| \\
    &\leq |s'p'| + (1+\eps)(|ss'| + |p't|) \\
        &\leq |x(s')-x(p')| + |y(s') - y(p')| + (1+\eps)(|x(s)-x(s')|+|x(p')-x(t)|) \\
        & \quad + (1+\eps)(|y(s)-y(s')|+|y(p')-y(t)|) \\
        &\leq |x(s')-x(p')| + (1+\eps)(|x(s)-x(s')|+|x(p')-x(t)|) \\
        & \quad + (1+\eps)(|y(s)-y(s')|+|y(s')-y(p')|+|y(p')-y(t)|).
    \end{align*}
    Using the above bounds, we obtain
    \begin{align*}
    |P(s,t)| &\leq |x(s')-x(p')| + (1+\eps)\left(|st|+\frac{6\alpha}{5}\, |sp|- |x(s')-x(p')|\right) + (1+\eps)\frac{46\alpha}{15}\, |sp| \\
    & \leq (1+\eps)|st| + (1+\eps)\left(\frac{64\alpha}{15}|sp|\right) - \eps \left(\frac{23}{40}|sp|\right) \\
    &\leq (1+\eps)|st| + \frac15 \left( (1+\eps)\frac{64\alpha}{3} -\frac{23\eps}{8}\right)|sp|\\
    &< (1+\eps)|st|,
    \end{align*}
    provided that $(1+\eps)\frac{64\alpha}{3} <\frac{23\eps}{8}$.
    Assuming $\eps<1$, this holds for all $\alpha\leq\eps/16$.
\end{proof}

The stretch analysis for the graph $H_2$ proceeds via a similar induction argument using the $L_1$-norm. The upper bounds established above for points within a cone $C$ may be adapted directly, but additional estimates on the $x$- and $y$-coordinates of points in the exclusion region $A(sp)$ are required. See \Cref{ssec:general} 
for details.

\subsection{Graph $H_2$ is a $(1+\eps)$-spanner}
\label{ssec:general}

Assume that the three $\Theta$-graphs, used in the construction of $H_0$, $H_0'$ and $H_0''$, jointly correspond to cones $\mathcal{C}=\{C_1,\ldots , C_{3k}\}$ of aperture $\alpha=2\pi/k$. The boundaries of the $3k$ cones with the same apex partition the plane into $3k$ cones of aperture $\frac{\alpha}{3}=2\pi/(3k)$. For $i=1,\ldots , 3k$, let $\widehat{C}_i$ denote the cone of aperture $\frac{\alpha}{3}$ with the same bisector as $C_i$. 

The following lemma plays a crucial role in the stretch analysis of $H_2$.

\begin{lemma}\label{lem:triple}
    Let $s_1p_1, sp\in E_{i,j}$ be such that $s_1\neq s$ and $s\in A(s_1p_1)$. 
    Assume further that the bisector of $C_i(s_1)$ is the positive $x$-axis.
  \begin{enumerate}[(a)]
        \item\label{lem:triple:a} For every point $t\in S\cap \widehat{C}(s)\setminus \{s\}$, we have $x(p_1)\leq x(t)$.
        \item\label{lem:triple:b} We have $|st|\geq \frac78\, |s_1p_1|$
    \end{enumerate}
\end{lemma}
\begin{proof}
Let $t\in S\cap \widehat{C}(s)\setminus \{s\}$; refer to \Cref{fig:triple}. We prove \Cref{lem:triple:a} by contradiction. Suppose that $x(t)<x(p_1)$. Recall that the interior of the triangle $\Delta(s_1p_1)$ is disjoint from $S$. This implies that $t\in \widehat{C}_i(s)\setminus C_i(s_1)$. In the remainder of the proof, we show that this gives $x(t)\leq \frac12 |s_1p_1|$.
Combined with $x(p)\leq x(t)$, we show that $|sp|<\frac34\, |s_1p_1|$, which contradicts the assumption that both $s_1p_1$ and $sp$ are in $E_{i,j}$. 

\begin{figure}[tbh]
        \centering
        \includegraphics[width=.8\textwidth]{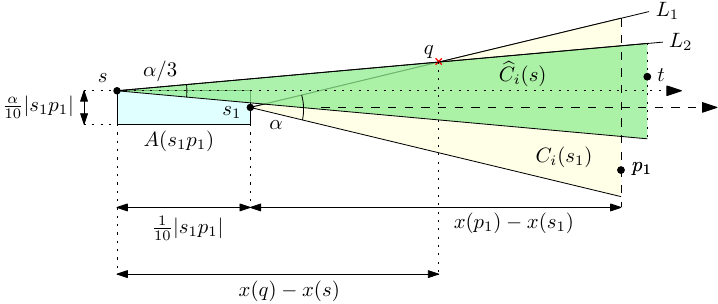}
        \caption{The position of cones $C_i(s_1)$ and $\widehat{C}_i(s)$, and the intersection point $q$ of their upper boundary rays}
        \label{fig:triple}
\end{figure}

We claim that for every $t\in \widehat{C}_i(s)\setminus C_i(s_1)$, we have $x(t)\leq \frac12 |s_1p_1|$. If the $x$-coordinate of $s$ decreases, then the set $\widehat{C}_i(s)\setminus C_i(s_1)$ monotonically increases. So we may assume w.l.o.g.\ that $s$ is on the left side of the rectangle $A(s_1p_1)$. We may further assume w.l.o.g.\ that $y(s)\geq (s_1)$. Then the rightmost point in $\widehat{C}_i(s)\setminus C_i(s_1)$ is the intersection point of the upper boundaries of $\widehat{C}_i(s)$ and $C_i(s_1)$, respectively, that we denote by $q$. Now $x(q)$ is increasing as $|y(s)-y(s_1)|$ increases. So we may assume w.l.o.g.\ that $s$ is the upper-left corner of $A(s_1p_1)$ (cf.~\Cref{fig:triple}). Under these assumptions, the upper boundaries of the cones $C_i(s_1)$ and $\widehat{C}_i(s)$ span the lines 
\[
\left \{
\begin{array}{ll}
        L_1: & y=\tan\frac{\alpha}{2} \cdot x \\
        L_2: & y= \tan\frac{\alpha}{6}\cdot (x+\frac{1}{10}|s_1p_1|)+ \frac{\alpha}{20}|s_1p_1| ,
    \end{array}
\right.
\]    
where $q=L_1\cap L_2$. The combination of the equations of the lines $L_1$ and $L_2$ yields 
\[ x(q)= \frac{\frac{1}{10} \tan \frac{\alpha}{6}+ \frac{\alpha}{20}}{\tan\frac{\alpha}{2}-\tan\frac{\alpha}{6}} |s_1p_1| 
< \frac{\left(\frac{1}{30}+\frac{1}{20}\right)\alpha |s_1p_1|}{\frac{\alpha}{2}-\frac{\alpha}{3}} 
=\frac{|s_1p_1|}{2},
\]
where we used the Taylor estimates $x<\tan x< 2x$ for $0<x<1$. 
Now $x(p)\leq x(t)\leq x(q)$ gives $x(p)\leq \frac12|p_1s_1|$. 
Point $p$ lies in the cone $C_1(s)$ of aperture $\alpha$. Consequently, 
\begin{align*}
 |y(p)-y(s)|& \leq \tan\frac{\alpha}{2}\cdot |x(p)-x(s)| \\
    &<\frac{\alpha}{2}\big(|x(p)-x(s_1)| + |x(s_1)-x(s)|\big)\\
    &\leq \frac{\alpha}{2}\,\left(\frac{1}{10}+\frac12\right) |s_1p_1|
    =\frac{3\alpha}{10}\, |s_1p_1|.
\end{align*}
By the triangle inequality and the assumption $\alpha<\frac12$, we obtain
\[
|sp|
< |x(s)-x(p)|+|y(s)-y(p)| 
< \left(\frac{1}{10}+\frac12\right)|s_1p_1| + \frac{3\alpha}{10}|s_1p_1| \\
 =\frac{6+3\alpha}{10}\, |s_1p_1|
< \frac34\, |s_1p_1|.
\]

It remains to prove \Cref{lem:triple:b}. 
If bisector of the cone $C_i(s_1)$ is the positive $x$-axis, then $x(s)\leq x(s_1)<x(p_1)\leq x(t)$. 
Taylor estimates (cf.~\Cref{eq:Taylor}) combined with $\alpha<\frac12$ yield
\[
|st| 
\geq  x(t)-x(s)
\geq x(p_1)-x(s_1)
\geq |s_1p_1|\cos \frac{\alpha}{2}
\geq \frac78\, |s_1p_1| .
\qedhere
\]
\end{proof}

\subparagraph{Coordinate system aligned with $st$.}
Let $s_1p_1, sp\in E_{i,j}$ such that $s_1\neq s$ and $s\in A(s_1p_1)$, as in \Cref{lem:triple}. 
However, we analyze the relative position of $s$, $t$, $s_1p_1$ and $s_1'p_1'$ in a coordinate system where $s$ is the origin and $\overrightarrow{st}$ is the positive $x$-axis; see \Cref{fig:stretchH2}. 

\begin{figure}[tbh]
        \centering
        \includegraphics[width=.85\textwidth]{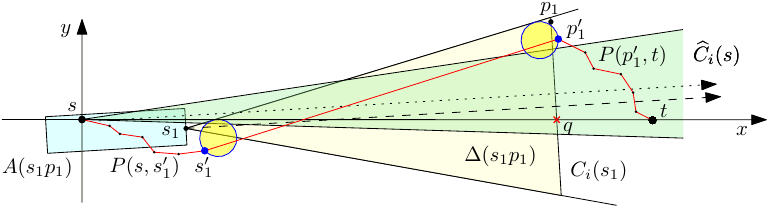}
        \caption{The concatenation of $P(s, s_1')$, $s_1'p_1'$ and $P(p_1', t)$}
        \label{fig:stretchH2}
\end{figure}

\begin{lemma}\label{lem:ssprime}
In the coordinate system where $s$ is the origin and $\overrightarrow{st}$ is the positive $x$-axis,
\begin{enumerate}
    \item $-\frac{\alpha}{20}|s_1p_1|\leq x(s_1)-x(s)\leq \left(\frac{1}{10}+\frac{\alpha}{20}\right)|s_1p_1| $;     and
    \item $|y(s_1)|\leq \frac{\alpha}{15}|s_1p_1|$.
\end{enumerate}
\end{lemma}
\begin{proof}
    Recall that $A(s_1p_1)$ is an $\frac{1}{10}|s_1p_1|\times \frac{\alpha}{10}|s_1p_1|$ rectangle, where $s\in A(s_1p_1)$ and $s_1$ is the midpoint of a shorter side of $A(s_1p_1)$. 
    It follows that 
        \[
-\frac{\alpha}{20}|s_1p_1|\leq x(s_1)-x(s)\leq |s s_1|\leq \left(\frac{1}{10}+\frac{\alpha}{20}\right)|s_1p_1| .
\]

Since the long sides of $A(s_1p_1)$ are parallel to the angle bisector of the cone $\widehat{C}_i(s)$, which contains $t$, then they make an angle at most $\alpha/6$ with the positive $x$-axis $\overrightarrow{st}$. Given that $y(s)=0$, this implies that $|y(s_1)|\leq \frac{\alpha}{20}|s_1p_1|+\frac{1}{10}|s_1p_1|\cdot \sin\frac{\alpha}{6} \leq (\frac{\alpha}{20}+\frac{\alpha}{60})|s_1p_1| = \frac{\alpha}{15}|s_1p_1|$, using Taylor estimates (cf.~\Cref{eq:Taylor}).
\end{proof}

\subparagraph{Projections of $s_1'$ and $p_1'$.} 
By \Cref{lem:triple}(a), $t\in C_i(s)$ lies beyond the triangle $\Delta(s_1p_1)$.
Let $\overrightarrow{r}$ be the ray parallel to $\overrightarrow{st}$ emanating from $s_1$. The cones $C_i(s_1)$ and $\widehat{C}_i(s)$ have parallel bisectors, but $\widehat{C}_i(s)$ has smaller aperture than $C_i(s_1)$, so $\overrightarrow{r}\subset C_i(s_1)$. \Cref{lem:xycoords} holds when $\overrightarrow{r}$ is the positive $x$-axis. We adapt \Cref{lem:xycoords} and \Cref{cor:spprime} for our setting, where $\overrightarrow{st}$ is the positive $x$-axis: 

\begin{lemma}\label{lem:xycoords+}  We have the following bounds on the $x$ and $y$-coordinates of $s_1'$ and $p_1'$
\begin{enumerate}
\item $\frac{-\alpha}{5}\, |s_1p_1| \leq x(s_1)-x(s_1') \leq (\frac14 + \frac{\alpha}{10})|s_1p_1|$ and $|y(s_1')-y(s_1)| \leq \frac{\alpha}{5}\, |s_1p_1|$.
\item $\frac78\, |s_1p_1| \leq x(p_1')-x(s_1) \leq (1 + \frac{\alpha}{5}) |s_1p_1|$ and $|y(p_1')-y(s_1)| < \frac{4\alpha}{3}\, |s_1p_1|$.
\item If $x(t) \leq x(p_1')$, then $x(p_1') - x(t) < \frac{2\alpha}{5}\, |s_1p_1|$, otherwise, $x(t) - x(p_1') \leq x(t)-x(s_1) - \frac78\, |s_1p_1|$.
\end{enumerate}
\end{lemma}
\begin{corollary}\label{cor:spprime+}
    $|x(s_1') - x(p_1')| >\frac{23}{40}\, |s_1p_1|$
\end{corollary}

The combination of \Cref{lem:ssprime} and \Cref{lem:xycoords+}(1) readily implies the following. 
\begin{corollary}\label{cor:xycoords}  We have
$\frac{-\alpha}{4}\, |s_1p_1| \leq x(s_1') \leq (\frac{7}{20} + \frac{3\alpha}{20})|s_1p_1|$ 
and $|y(s_1')| < \frac{\alpha}{4}\, |s_1p_1|$.
\end{corollary}

\begin{lemma}\label{lem:inductionhyp+}
If $\alpha < \frac12$, then the segment $ss_1'$ satisfies $|ss_1'| < |st|$ and the segment $p_1't$ satisfies $|p_1't| < |st|$.
\end{lemma}
\begin{proof}
The Euclidean length of the segment $ss_1'$ is bounded above by its $L_1$ norm. 
\Cref{cor:xycoords}, combined with $\alpha < \frac12$ and \Cref{lem:triple}(b), implies that
\begin{align*}
      |ss_1'| &\leq |x(s_1')| + |y(s_1')| 
      \leq \left(\frac{7}{20}+\frac{3\alpha}{20}\right)|s_1p_1| + \frac{\alpha}{4}\, |s_1p_1| \\
      &= \left(\frac{7}{20}+\frac{2\alpha}{5}\right)|s_1p_1|
      <\frac{13}{20}\, |s_1p_1| <\frac78\, |s_1p_1| <|st| .
\end{align*}
    For the segment $p_1't$, \Cref{lem:inductionhyp} showed that $|p_1't|<\frac78 |s_1p_1|$. 
    Combined with \Cref{lem:triple}(b), this further implies $|p_1't|< |st|$, as claimed. 
\end{proof}

    

In the remainder of \Cref{sec:stretch}, we prove the main result of \Cref{ssec:general}

\begin{lemma}\label{lem:stretch:H2}
    If $\alpha\leq \eps/16$, then $H_2$ is a $(1+\eps)$-spanner for $S$. 
\end{lemma}
\begin{proof}
We claim that for every pair of points $s,t\in S$, the graph $H_2$ contains an $st$-path of length at most $(1+\eps)|st|$. We prove the claim by induction on the length $|st|$. In the base case, we have $s=t$, and the claim trivially holds. Assume that we are given a point pair $s,t\in S$, where $s\neq t$, and the claim holds for all pairs $\{s^*,t^*\}\subseteq  S$ with $|s^*t^*|<|st|$. 

Let $i\in \{1,\ldots ,3k\}$ such that $t\in \widehat{C}_i(s)$. Then $t\in C_i(s)$, and the $\Theta$-graph has an edge $sp$ corresponding to the cone $C_i(s)$, where $sp\in E_{i,j}$ for some $j\in \mathbb{Z}$. Either $sp$ corresponds to an edge $s'p'\in E^*_{i,j}(H_2)$, or the edge $sp$ was pruned because $s\in A(s_1p_1)$ for some edge $s_1p_1\in E_{i,j}$, and we have $s_1'p_1'\in E^*_{i,j}(H_2)$.
In the first case, the induction step can be completed as in the proof of \Cref{lem:stretchH0}, but it also reduces to the second case by setting $s=s_1$. It suffices to analyze the second case. 

For the sake of the stretch analysis, we use a coordinate system where $s$ is the origin, and the ray $\overrightarrow{st}$ is the positive $x$-axis; see \Cref{fig:stretchH2}. We can now outline the remainder of the proof: We show that $|s s_1'|<|st|$ and $|p_1't|<|st|$,  and by the induction hypothesis $H_2$ contains paths $P(s,s_1')$ and $P(p_1',t)$ such that 
$|P(s,s_1')|\leq (1+\eps)\, |ss_1'|$ and $|P(p_1',t)|\leq (1+\eps)\, |p_1't|$.
We obtain an $st$-path as a concatenation $P(s,t)=P(s,s_1')\oplus s_1'p_1'\oplus P(p_1',t)$,
and we can bound the length of $P(s,t)$ as 
\begin{equation}\label{eq:initial}
    |P(s,t)| =  |P(s,s_1')| + |s_1' p_1'| + |P(p_1',t)| 
            \leq (1+\eps) \Big( |ss_1'| + |p_1't|\Big) + |s_1' p_1'| .
\end{equation}
The length of every line segment $ab$ is bounded above by its $L_1$-norm, that is, $|ab|\leq |x(a)-x(b)| + |y(a)-y(b)|$. Since we assume that $s$ is the origin and $t$ is on the positive $x$-axis, then 
$x(s)=y(s)=0$, $x(t)=|x(t)-x(s)|=|st|$, and $y(t)=0$. 

By \Cref{lem:inductionhyp+} and the induction hypothesis, $H_2$ contains paths $P(s, s_1')$ and $P(p_1', t)$ satisfying $|P(s,s_1')| \leq (1+\eps)|ss'_1|$ and $|P(p_1', t)| \leq (1+\eps)|p'_1t|$. 

    \subparagraph{Upper bounds for distances in $x$-projections.} We compute the total $x$-distance traveled along the path. Similar to the stretch analysis of $H_0$, the path can backtrack if $x(s_1') < x(s)$ or $x(p_1') > x(t)$. By \Cref{lem:xycoords+} and \Cref{lem:inductionhyp+}, the backtracking can be bounded as
    \[ |x(s) - x(s_1')| + |x(s_1')-x(p_1')| + |x(p_1')-x(t)| 
    \leq |st| + 2\left(\frac{\alpha}{4}\,|s_1p_1| + \frac{2\alpha}{5}\,|s_1p_1|\right) 
    = |st| + \frac{13\alpha}{10}\,|s_1p_1|.\]
    In particular,
    \[|x(s) - x(s_1')| + |x(p_1') - x(t)| \leq |st| + \frac{13\alpha}{10}\,|s_1p_1| - |x(s_1')-x(p_1')|.\]

    \subparagraph{Upper bounds for distances in $y$-projections.} Using the bounds in \Cref{lem:xycoords+} and the fact that $y(s)=y(t)=0$, the triangle inequality yields
    \begin{align*}
    |y(s) - y(s_1')|+ |y(s_1') - y(p_1')| + |y(p_1') - y(t)| 
    &\leq 2|y(s_1')| + 2|y(p_1')|\\
    &\leq 2\Big( |y(s_1')| + |y(p_1')-y(s_1)| + |y(s_1)|\Big)\\
    &\leq 2\left( \frac{\alpha}{4}\,|s_1p_1| + \frac{4\alpha}{3}\, |s_1p_1| +  \frac{\alpha}{15}\, |s_1p_1|\right) \\
    &= \frac{33\alpha}{10}\, |s_1p_1|.
    \end{align*}

    \subparagraph{Length of $P(s,t)$.} Now, we combine the above results in order to bound the length of the path $P(s,t)$. By the induction hypothesis,
    \begin{align*}
    |P(s,t)| &= |P(s, s_1')| + |s_1'p_1'| + |P(p_1', t)| \\
    &\leq |s_1'p_1'| + (1+\eps)(|ss_1'| + |p_1't|)  \\
    &\leq |x(s_1')-x(p_1')| + |y(s_1') - y(p_1')| + (1+\eps)(|x(s)-x(s_1')|+|x(p_1')-x(t)|) \\
        & \quad + (1+\eps)(|y(s)-y(s_1')|+|y(p_1')-y(t)|) \\
        &\leq |x(s_1')-x(p_1')| + (1+\eps)(|x(s)-x(s_1')|+|x(p_1')-x(t)|) \\
        & \quad + (1+\eps)(|y(s)-y(s_1')|+|y(s_1')-y(p_1')|+|y(p_1')-y(t)|).
    \end{align*}
    Using the above bounds, we obtain
    \begin{align*}
    |P(s,t)| &\leq |x(s_1')-x(p_1')| + (1+\eps)\left(|st|+\frac{13\alpha}{10}\, |s_1p_1|- |x(s_1')-x(p_1')|\right) + (1+\eps)\frac{33\alpha}{10}\, |s_1p_1| \\
    & \leq (1+\eps)|st| + (1+\eps)\left(\frac{23\alpha}{5}|s_1p_1|\right) - \eps \left(\frac{23}{40}|s_1p_1| \right)\\
    &\leq (1+\eps)|st| + \left( (1+\eps)\frac{23\alpha}{5} -\frac{23\eps}{40}\right)|s_1p_1|\\
    &\leq (1+\eps)|st| + \frac{23}{5}\left( (1+\eps)\alpha -\frac{\eps}{8}\right)|s_1p_1|\\
    &< (1+\eps)|st|,
    \end{align*}
    provided that $(1+\eps)\alpha <\frac{\eps}{8}$.
    Assuming $\eps<1$, this holds for all $\alpha\leq \eps/16$.
\end{proof}

\section{Crossing Analysis}
\label{sec:crossing}

We first show that the lengths of any two crossing edges in the graph $H_1$ 
differ by a factor of $O(1/\eps^2)$. Both $H_0$ and $H_2$ are subgraphs of $H_1$, so this property holds for $H_0$ and $H_2$, too. 

\begin{lemma} \label{lem:ratio}
If edges $e_1,e_2\in E(H_1)$ cross and $|e_1|\leq |e_2|$, then $|e_2|/|e_1|\leq O(1/\eps^2)$. 
\end{lemma}
\begin{proof}
Let $ab$ and $cd$ be two crossing edges in $H_1$ such that $|ab|\leq |cd|$. They correspond to empty regions $R_{ab}$ and $R_{cd}$ such that $a,b\in \partial R_{ab}$ and $c,d\in \partial R_{cd}$. Furthermore $R_{ab}$ is the convex hull of two congruent regular $3k$-gons $D_a$ and $D_b$, where $a\in \partial D_a$ and $b\in \partial D_b$. Similarly, $R_{cd}=\conv(D_c\cup D_d)$, where $c\in \partial D_c$ and $d\in \partial D_d$. 

Let $s_2p_2$ be the edge of a $\Theta$-graph corresponding to $cd$ (that is, $cd=s_2'p_2'$). Then  $\diam(D_c)=\frac{\alpha}{5}\,|s_2p_2|$, where $\eps/32\leq \alpha \leq \eps/16$, and we have $|cd|\geq \Omega(|s_2p_2|)$ by \Cref{cor:spprime}.
In particular, $\diam(D_c)\geq \Omega (\eps\, |cd|)$ and the distance between $D_c$ and $D_d$ is also $\Omega(|cd|)$.
Let $\lambda$ denote the edge length of the regular $3k$-gon $D_c$. Since $D_c$ is a regular $3k$-gon, we have $\lambda = \diam(D_c)\sin(\pi/(3k)) =\Theta(\eps^2 |cd|)$.  We claim that $|ab|\geq \lambda/3 = \Omega(\eps^2 |cd|)$. 

\begin{figure}[tbh]
        \centering
        \includegraphics[width=\textwidth]{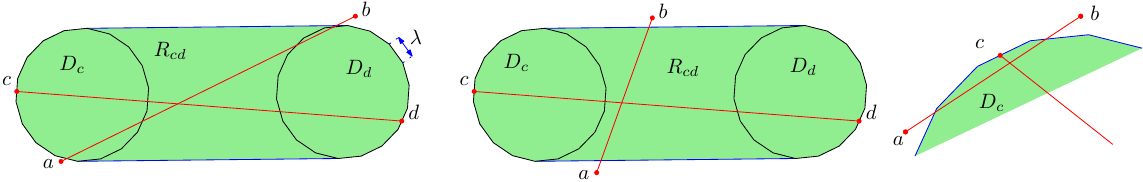}
        \caption{Edge $ab$ crosses $R_{cd}$. Left: $ab$ intersects both $D_c$ and $D_d$. Middle: $ab$ intersects both outer common tangents of $D_c$ and $D_d$. Right: $ab$ crosses two nonadjacent sides of $D_c$.}
        \label{fig:crossings1}
\end{figure}

Suppose, to the contrary, that $|ab|<\lambda/3$.
Since $cd\subset R_{cd}$ but $c,d\notin \mathrm{int}(R_{cd})$, then the intersection point $ab\cap cd$ lies in $\mathrm{int}(R_{cd})$. As $a,b\notin \mathrm{int}(R_{cd})$, then the edge $ab$ crosses the empty region $R_{cd}$. Points $c$ and $d$ partition the boundary $\partial R_{cd}$ into two arcs, each containing a common outer  tangent segment between $\partial D_c$ and $\partial D_d$. Since $ab$ crosses $cd$, then $ab$ must intersect both arcs. If $ab$ intersects both $D_c$ and $D_d$, then $|ab|$ is at least the distance between $D_c$ and $D_d$, which is more than $\lambda$ (\Cref{fig:crossings1}, left). If $ab$ intersects both outer tangents of $D_c$ and $D_d$, then $|ab|$ is at least the width of the tube $U(D_c,D_d)$, which is bounded below by the diameter of the inscribed circle of $D_c$, which is in turn much larger than $\lambda$  (\Cref{fig:crossings1}, middle). Therefore, we may assume that $ab$ intersects exactly one of $D_c$ and $D_d$. Assume w.l.o.g.\ that $ab$ intersects $D_c$. 
If $ab$ intersects two nonadjacent sides of $D_c$, then $|ab|$ is at least the distance between those sides, which is at least $\lambda$  (\Cref{fig:crossings1}, right). Hence, we may assume that $ab$ intersects two consecutive sides of $D_c$. Let $v$ be the common endpoint of the two sides of $D_c$ crossed by $ab$; see \Cref{fig:crossings2}. 

Note that $v\notin R_{ab}$, otherwise we would have $\conv \{a,b,v\}\subset R_{ab}$ by convexity, and  the edge $cd$ would not cross $ab$. This further implies that $D_a$ (resp., $D_b$) intersects at most one side of $D_c$: Indeed, as $D_a$ (resp., $D_b$) is homothetic to $D_c$, if $D_a$ intersects two consecutive sides of $D_c$, it would contain the common endpoint of the two sides. 

\begin{figure}[tbh]
        \centering
        \includegraphics[width=.95\textwidth]{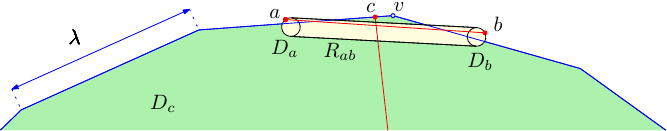}
        \caption{Edge $ab$ crosses two adjacent sides of $D_c$.}
        \label{fig:crossings2}
\end{figure}

Let $s_1p_1$ be the edge of the $\Theta$-graph corresponding to $ab$ (that is, $ab=s_1'p_1'$), 
and w.l.o.g.\ edge $s_1p_1$ lies in the cone $C_i(s_1)$. We distinguish between two cases described in \Cref{lem:cases}:

\begin{itemize}
\item \textbf{Case~1:} The boundary lines of the tube $U(D_a,D_b)$ are parallel to a ray on the boundary of $C_i(s_1)$. By construction, each boundary ray of $C_i$ is parallel to two opposite sides of $D_c$. 
    Consequently, every side of $R_{ab}$ is parallel to a side of $D_c$. Let $T_1$ and $T_2$ be the halfplanes that contain $R_{ab}$, bounded by lines $\partial T_1$ and $\partial T_2$ which contain the two sides of $R_{ab}$ parallel to the two adjacent sides of $D_c$ incident to $v$. 
    Then we have $D_c\subset T_1\cap T_2$, which in turn implies $v\in T_1\cap T_2$. 
    Since $\partial T_1$ and $\partial T_2$ are parallel to two consecutive sides of $R_{ab}$, then $\partial T_1\cap \partial T_2$ is a vertex of $R_{ab}$, and $v\in R_{ab}$: A contradiction. 
%
\item \textbf{Case~2:} We have $b=p_1=p_1'$ and $D_b\subset \Delta(s_1p_1)$. Recall that $\Delta(s_1p_1)$ is an empty triangle in the cone $C_i(s_1)$, and the three sides of $\Delta(s_1p_1)$ are parallel to three sides of $D_c$~\cite[Lemma~1]{BoseCMRV16}. Let $T$ be the halfplane bounded by a line in $\partial U(D_a,D_b)$ that contains $U(D_a,D_b)$ (hence $R_{ab}$), but does not contain $v$.  \Cref{lem:termination} implies $s_1,p_1\in U(D_a,D_b)$, and so $s_1p_1\subset T$. \Cref{lem:construction:tubegeneric} yields $|s_1p_1|<2\, |ab|<\frac23\, \lambda$. 
Hence $s_1p_1$ intersects the same two sides of $D_c$ as $s_1'p_1'$.

    We can now proceed similarly to Case~1, arguing about the segment $s_1p_1$ and the empty triangle $\Delta(s_1p_1)$ in place of $ab=s_1'p_1'$ and $R_{ab}$. Let $T_1$ and $T_2$ be the halfplanes that contain $\Delta(s_1p_1)$, bounded by lines $\partial T_1$ and $\partial T_2$ containing the sides of $\Delta(s_1p_1)$ that cross the two sides of $D_c$ incident to $v$ (possibly, $T_1=T_2$). 
    By convexity, we have $D_c\subset T_1\cap T_2$.
    By \Cref{lem:parallel}, the lines $\partial T_1$ and $\partial T_2$ are each parallel to some sides of $D_c$, and so $s_1p_1\subset T_1$ and $s_1p_1\subset T_2$ imply $v\in T_1$ and $v\in T_2$, respectively. 
     We conclude that $v\in T_1\cap T_2$, hence $v\in \Delta(s_1p_1)$. By convexity, we have $\conv \{a,b,v\}\subset R_{ab}$, which further implies $c\in \Delta(s_1p_1)$. Since the sides of the $3k$-gon (i.e., the cones of the $\Theta$-graphs) are in general position with respect to $S$, then $c\in \mathrm{int}(\Delta(s_1p_1))$, contradicting the assumption that $\Delta(s_1p_1)$ is an empty triangle. 
\end{itemize}

Both cases lead to contradiction. We conclude that $|ab|\geq \lambda/3 = \Theta(\eps^2|cd|)$.
\end{proof}

%

\subparagraph{Bounding the number of crossings.} 
Fix an edge $e\in E(H_2)$. We analyze the number of crossings between $e$ and edges in $E_{i,j}(H_2)$ for fixed indices $i,j$. For convenience, we assume that the bisector of $C_i$ is the positive $x$-axis (our arguments extend to all $i$ after suitable rotation).
Recall that every edge $s'p'\in E_{i,j}(H_2)$ corresponds to some edge $sp\in E^*_{i,j}$ of a $\Theta$-graph, which in turn corresponds to the cone $C_i$. We define the axis-aligned rectangle $M_j=[-\frac{1}{30}\, \varphi^{j}, (1+\eps)\, \varphi^j]\times [-\eps\, \varphi^j,\eps\,\varphi^j]$. It is clear that $\area(M_j)=\Theta(\eps\, \varphi^{2j})$. For every edge $s'p'\in E_{i,j}(H_2)$, let $M_j(s'p')=M_j+(x(s),y(s))$ be a translate of $M_j$; see \Cref{fig:boundingbox}.

\begin{figure}[tbh]
        \centering
        \includegraphics[width=.7\textwidth]{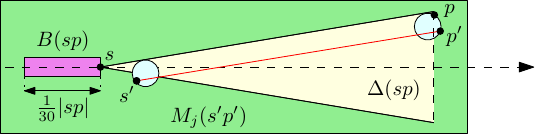}
        \caption{The rectangle $M_j(s'p')$ contains both $B(sp)$ and $s'p'$.}
        \label{fig:boundingbox}
\end{figure}

\begin{lemma}\label{lem:boundingbox}
    For every edge $s'p'\in E_{i,j}(H_2)$, we have $B(sp)\subset M_j(s'p')$ and $s'p'\subset M_j(s'p')$.
\end{lemma}
\begin{proof}
Let $s'p'\in E_{i,j}(H_2)$ correspond to the $\Theta$-graph edge $sp$ with $\varphi^{j-1}\leq |sp|<\varphi^{j}$. Assume w.l.o.g.\ that the bisector of $C_i(s)$ is the positive $x$-axis and $s$ is the origin. In this coordinate system, $B(sp)=[-\frac{1}{30}|sp|,0]\times [-\frac{\alpha}{60}|sp|,\frac{\alpha}{60}|sp|]$. Since $|sp|<\varphi^j$, then $B(sp)\subset M_j(s'p')$.

To prove that $s'p' \subset M_j(s'p')$, it is enough to show that $s',p'\in M_j(s'p')$ by convexity. We apply \Cref{lem:xycoords}(1--2) for the same coordinate system (with a hypothetical point $t$ on the bisector of $C_i(s)$). It gives $\frac{-\alpha}{5}\, |sp| \leq x(s') \leq (\frac14 + \frac{\alpha}{10})|sp|$ and $|y(s')| \leq \frac{\alpha}{5}\, |sp|$. Combined with $\alpha\leq \eps/16$ and $|sp|<\varphi^j$, this implies that $s\in M_j(s'p')$. Similarly, \Cref{lem:xycoords}(2) gives $\frac78\, |sp| \leq x(p') \leq (1 + \frac{\alpha}{5}) |sp|$ and $|y(p')| < \frac{4\alpha}{3}\, |sp|$, hence $p'\in M_j(s'p')$. 
\end{proof}
\begin{figure}[tbh]
        \centering
        \includegraphics[width=.95\textwidth]{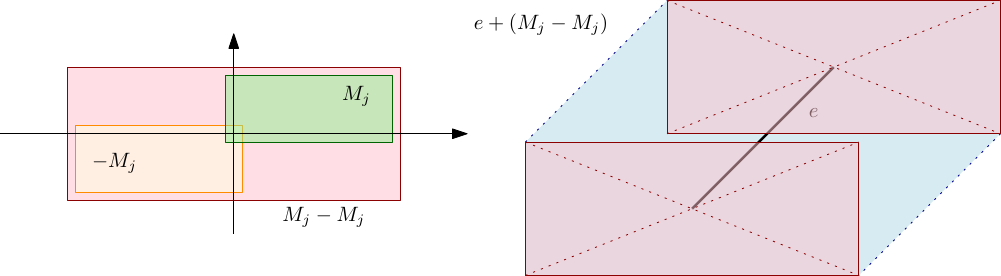}
        \caption{Left: Axis-aligned rectangles $M_j$ and $M_j-M_j$ in a Cartesian coordinate system. Right: A line segment $e$ and the Minkowski sum $e+(M_j-M_j)$.}
        \label{fig:minkowski}
\end{figure}
We show that for every edge $s'p'\in E_{i,j}(H_2)$ that crosses $e$, the box $B(sp)$ lies in a Minkowski sum $e+(M_j-M_j)$; see~\Cref{fig:minkowski}.
\begin{lemma}\label{lem:minkowski}
    If an edge $s'p'\in E_{i,j}(H_2)$ crosses a line segment $e$, then $B(sp)\subset e+(M_j-M_j)$.
\end{lemma}
\begin{proof}
  Let $q=e\cap s'p'$ be the intersection point. By \Cref{lem:boundingbox}, the rectangle $M_j(s'p') = M_j+(x(s),y(s))$ contains $B(sp)$ as well as points $q$ and $s$. 
  This means that the vector $\overrightarrow{sq}$ lies in $M_j$, and for every point $b\in B(sp)$, vector $\overrightarrow{sb}$ lies in $M_j$. For all $b\in B(sp)$, we have $b=q+\overrightarrow{qb} = q+\overrightarrow{sb}-\overrightarrow{sq}\subset e+(M_j-M_j)$, hence $B(sp)\subseteq e+(M_j-M_j)$, as claimed.
\end{proof}
We can now use volume arguments---comparing the area of $e+(M_j-M_j)$ with the area of $B(sp)$ for $s'p'\in E_{i,j}(H_2)$. We distinguish between two cases: \Cref{lem:long} handles edges in $E_{i,j}(H_2)$ of length at least $\width(R_e)$, and \Cref{lem:short} is about shorter edges. 

\begin{lemma} \label{lem:long}
Let $e\in E(H_2)$, $i\in \{1,\ldots ,3k\}$, and $j\in \mathbb{Z}$.
If $2\cdot \varphi^j\geq \width(R_e)$, then $e$ crosses $O(1/\eps^2)$ edges in $E_{i,j}(H_2)$. 
\end{lemma}
\begin{proof}
The Minkowski sum $M_j-M_j$ is an axis-aligned rectangle of size $\Theta(\varphi^j)\times \Theta(\eps\, \varphi^j)$, hence $\area(M_j-M_j)=O(\eps\, \varphi^{2j})$. The area of $e+(M_j-M_j)$ is bounded above by 
    \[
    \area(e+(M_j-M_j))\leq \frac{|e|}{\width(M_j-M_j)}\cdot \area(M_j-M_j)
    \leq O\left(\frac{|e|}{\eps\, \varphi^j}\, \eps \varphi^{2j}\right)
    =O(\varphi^j\, |e|).
    \]

For every edge $s'p'\in E_{i,j}(H_2)$ that crosses $e$, we have $B(sp)\subset e+(M_j-M_j)$ by \Cref{lem:minkowski}, and $\area(B(sp))=\Theta(\eps\, \varphi^{2j})$. The rectangles $B(sp)$ are pairwise disjoint by \Cref{lem:packing}. Consequently, the number of edges in $E_{i,j}(H_2)$ that cross $e$ is at most 
\[
\frac{\area(e+(M_j-M_j))}{\Theta(\eps \varphi^{2j})} 
=O\left( \frac{\varphi^j\, |e|}{\eps\, \varphi^{2j}}\right)
=O\left(\frac{|e|}{\eps\, \varphi^j}\right).
\]

Since $\width(R_e)=\Theta(\eps\, |e|)$, our assumption $\width(R_e)\leq 2\cdot \varphi^{j}$ gives $|e|\leq O(\varphi^j/\eps)$. So the number of edges that cross $e$ is $O(|e|/(\eps\, \varphi^j)) \leq  O\big((\varphi^j/\eps)/(\eps \varphi^j)\big)= O(1/\eps^2)$. 
\end{proof}

\begin{lemma} \label{lem:short}
Let $e\in E(H_2)$, $i\in \{1,\ldots ,3k\}$, and $j\in \mathbb{Z}$.
If $2\cdot \varphi^j< \width(R_e)$, then $e$ crosses  $O(1/\eps^2)$ edges in $E_{i,j}(H_2)$. 
\end{lemma}
\begin{proof}
Assume $e = ab$. Recall that $R_e$ is the convex hull of two congruent regular $3k$-gons, say $D_a$ and $D_b$, such that $a\in \partial D_a$ and $b\in \partial D_b$. For every edge $s'p'\in E_{i,j}(H_2)$, we have $|s'p'|< 2\, |sp| \leq 2\cdot \varphi^j< \width(R_e)$ by \Cref{lem:xycoords}. So $s'p'$ cannot cross both common outer tangent lines of $D_a$ and $D_b$. If $s'p'$ crosses edge $e$, then it must intersect $D_a$ or $D_b$.

Assume w.l.o.g.\ that an edge $s'p'\in E_{i,j}(H_2)$ intersects $D_a$. Let $q=ab\cap s'p'$ and $r\in D_a\cap s'p'$. Since $a\in \partial D_a$, then $|ar|\leq \diam(D_a)$ and $|qr|\leq |s'p'|<\diam(D_a)$. The triangle inequality yields $|aq|\leq |ar|+|rq|\leq 2\, \diam(D_a)$. 
Similarly, if $s'p'\in E_{i,j}(H_2)$ intersects $D_b$ and crosses $e$ at $q=ab\cap s'p'$, then $|bq|\leq 2\,\diam(D_b)=2\,\diam(D_a)$. In both cases, $s'p'$ crosses a subsegment of $e$ of length  $2\,\diam(D_a)=\Theta(\eps |e|)$ incident to $a$ or $b$. 

Let $e_a$ (resp., $e_b$) be a subsegment of $e$ of length $2\,\diam(D_a)$ incident to $a$ (resp., $b$). As noted above, if $s'p'\in E_{i,j}(H_2)$ crosses $e$, it also crosses $e_a$ or $e_b$.
By \Cref{lem:minkowski}, we have $B(sp)\subset e_a+(M_j-M_j)$ or $B(sp)\subset e_b+(M_j-M_j)$.
As the segments $e_a$ and $e_b$ are translates of each other, then $\area(e_a+(M_j-M_j))=\area(e_b+(M_j-M_j)) = 
 \frac{|e_a|}{\width(M_j-M_j)}\cdot \area(M_j-M_j) = O(\varphi^j\cdot \eps\, |e|)$.
Since the rectangles $B(sp)$ are pairwise disjoint by \Cref{lem:packing}, the number of edges in $E_{i,j}(H_2)$ that cross $e$ is at most 
\[
\frac{2\, \area(e_a+(M_j-M_j))}{\Theta(\eps \varphi^{2j})} 
=O\left(\frac{\varphi^j\cdot \eps\, |e|}{\eps\, \varphi^{2j}}\right)
=O\left(\frac{|e|}{\varphi^j}\right).
\]
By \Cref{lem:ratio}, we have $|s'p'|\geq \Omega(\eps^2 \, |e|)$, hence 
$|e|\leq O(\varphi^j/\eps^2)$. So the number of crossing edges is bounded by $O(|e|/ \varphi^j) \leq  O\big((\varphi^j/\eps^2)/\varphi^j\big)= O(1/\eps^2)$, as required. 
\end{proof}

\begin{lemma}\label{lem:crossings}
    Every edge in $E(H_2)$ has $O(\eps^{-3}\log \eps^{-1})$ crossings. 
\end{lemma}
\begin{proof}
    Let $e\in E(H_2)$. Denote by $J$ the set of indices $j\in \mathbb{Z}$ such that some edge $f\in E_{i,j}(H_2)$ for some $i\in \{1,\ldots,3k\}$ crosses $e$. \Cref{lem:ratio} yields $\Omega(\eps^2\, |e|)\leq |f|\leq O(|e|/\eps^2)$.  In particular, if $f\in E_{i,j}(H_2)$, then $\Omega(\eps^2\, |e|)\leq \varphi^j \leq O(|e|/\eps^2)$. Consequently, $\max J-\min J \leq  O(\log \eps^{-4}) =O(\log \frac{1}{\eps})$, hence $|J|=O(\log \frac{1}{\eps})$.
    By \Cref{lem:long,lem:short}, edge $e$ crosses $O(\eps^{-2})$ edges in $E_{i,j}(H_2)$ for every $i,j\in \mathbb{Z}$. Summation over all $i=1,\ldots ,3k = O(1/\eps)$ and all $j\in J$, implies that the total number of edges that cross $e$ is $O(\eps^{-3}\log \eps^{-1})$, as claimed. 
\end{proof}

\begin{corollary}\label{cor:crossings}
   The $(1+\eps)$-spanner $H_2$ has $O(\eps^{-4}\log \eps^{-1}\cdot n)$ crossings, where $n=|S|$. 
\end{corollary}
\begin{proof}
    A $\Theta$-graph with apex angle $\alpha =\Theta(\eps)$ has $O(n/\eps)$ edges. We have constructed $H_2$ such that each edge in $E(H_2)$ corresponds to an edge of three such $\Theta$-graphs. Consequently, $|E(H_2)|=O(n/\eps)$. Each edge $e\in E(H_2)$ has $O(\eps^{-3}\log \frac{1}{\eps})$ crossings by \Cref{lem:crossings}. Summation over all edges $e\in E(H_2)$ shows that the total number of crossings is $O(\eps^{-4}\log \frac{1}{\eps}\cdot n)$.
\end{proof}

\section{Outlook}
\label{sec:con}

We have proved existential upper and lower bounds on (i) the total number of crossings, (ii) the maximum number of crossings per edge, and (iii) the maximum ratio between the lengths of crossing edges in a $(1+\eps)$-spanner for $n$ points in the plane. 
Closing the gaps between these upper and lower bounds is an obvious open problem. 
The corresponding optimization problems remain open, as well. 
It is NP-hard to decide whether a point set $S\subset \mathbb{R}^2$ admits a noncrossing $t$-spanner with at most $K$ edges for given $S$, $K$, and $t$~\cite{GiannopoulosKKKM10} (see also~\cite{CarmiC13}).
Is it also NP-hard to decide whether a noncrossing $t$-spanner exists for a given  $S\subset \mathbb{R}^2$ and $t>1$ \cite[Problem~5]{BoseS13}? Approximation algorithms for geometric $t$-spanners with the minimum number of crossings (or min-max number of crossings per edge)
are left for future research.  

\bibliography{spanners}

\end{document}